\documentclass[a4paper,11pt,reqno]{amsart}
\usepackage[latin1]{inputenc}
\usepackage{mathrsfs}
\usepackage{dsfont}
\usepackage{hyperref}
\usepackage{amsmath}
\usepackage{amssymb}
\usepackage{amsthm}
\usepackage{amsfonts}
\usepackage{amstext}
\usepackage{amsopn}
\usepackage{amsxtra}
\usepackage{mathrsfs}
\usepackage{dsfont}
\usepackage{esint}
\usepackage{enumitem}

\newtheorem{theorem}{Theorem}
\newtheorem{lemma}{Lemma}
\newtheorem{proposition}{Proposition}
\newtheorem{corollary}{Corollary}
\newtheorem{remark}{Remark}
\newcommand{\R}{\mathbb{R}}
\newcommand{\Z}{\mathbb{Z}}
\newcommand{\C}{\mathbb{C}}
\newcommand{\bP}{\mathbb{P}}

\newcommand{\dps}{\displaystyle}
\newcommand{\ii}{\infty}
\newcommand\1{{\ensuremath {\mathds 1} }}
\renewcommand\phi{\varphi}
\newcommand{\gH}{\mathfrak{H}}
\newcommand{\gS}{\mathfrak{S}}

\newcommand{\cT}{\mathcal{T}}

\newcommand{\cC}{\mathcal{C}}

\newcommand{\cE}{\mathcal{E}}

\newcommand{\norm}[1]{ \left\| #1 \right\|}
\newcommand{\tr}{{\rm Tr}\,}

\renewcommand{\geq}{\geqslant}
\renewcommand{\leq}{\leqslant}

\renewcommand{\tilde}{\widetilde}

\newcommand{\be}{\begin{equation}}
\newcommand{\ee}{\end{equation}}
\newcommand{\bq}{\begin{equation}}
\newcommand{\eq}{\end{equation}}

\newcommand{\eps}{\varepsilon}
\usepackage{color}

\newcommand{\nn}{\nonumber}
\newcommand{\bDelta}{\mathbf{\Delta}}

\title{The Local Density Approximation in Density Functional Theory}

\author[M. Lewin]{Mathieu Lewin}
\address{CNRS \& CEREMADE, Universit\'e Paris-Dauphine, PSL University, 75016 Paris, France} 
\email{mathieu.lewin@math.cnrs.fr}

\author[E.H. Lieb]{Elliott H. Lieb}
\address{Departments of Mathematics and Physics, Jadwin Hall, Princeton University, Washington Rd., Princeton, NJ 08544, USA}
\email{lieb@princeton.edu}

\author[R. Seiringer]{Robert Seiringer}
\address{IST Austria (Institute of Science and Technology Austria), Am Campus 1, 3400 Klosterneuburg, Austria}
\email{robert.seiringer@ist.ac.at}

\date{\today}

\begin{document}
 
 \begin{abstract}
We give the first mathematically rigorous justification of the Local Density Approximation in Density Functional Theory. We provide a quantitative estimate on the difference between the grand-canonical Levy-Lieb energy of a given density (the lowest possible energy of all quantum states having this density) and the integral over the Uniform Electron Gas energy of this density. The error involves gradient terms and justifies the use of the Local Density Approximation in the situation where the density is very flat on sufficiently large regions in space.

\bigskip

\noindent \sl \copyright~2019 by the authors. This paper may be reproduced, in its entirety, for non-commercial purposes. Final version to appear in \emph{Pure and Applied Analysis}.
 \end{abstract}

 \maketitle

 \tableofcontents

 \newpage
\section{Introduction}
 
Density Functional Theory (DFT)~\cite{DreGro-90,ParYan-94,EngDre-11,BurWag-13,PriGroBur-15} is the most efficient approximation of the many-body Schrödinger equation for electrons. It is used in several areas of physics and chemistry and its success in predicting the electronic properties of atoms, molecules and materials is unprecedented. Among the many functionals that have been developed over the years~\cite{MarHea-17}, the \emph{Local Density Approximation} (LDA) is the standard and simplest scheme~\cite{HohKoh-64,KohSha-65,DreGro-90,ParYan-94,PerKur-03}. It is not as accurate as its successors involving gradient corrections, but it is considered as \emph{``the mother of all approximations''}~\cite{PerSch-01} and it is still one of the methods of choice in solid state physics. 

In the orbital-free formulation of Density Functional Theory~\cite{Levy-79,Lieb-83b}, the Local Density Approximation consists in replacing the full ground state energy by a local functional as follows:
\begin{equation}
 F_{\rm LL}(\rho)\approx\frac{1}{2}\int_{\R^3}\int_{\R^3}\frac{\rho(x)\rho(y)}{|x-y|}\,dx\,dy+\int_{\R^3}e_{\rm UEG}\big(\rho(x)\big)\,dx.
 \label{eq:LDA}
\end{equation}
Here $\rho$ is the given one-particle density of the system and $F_{\rm LL}(\rho)$ is the \emph{Levy-Lieb functional}~\cite{Levy-79,Lieb-83b}, the main object of interest in DFT. This is the lowest possible Schrödinger energy of all quantum states having the prescribed density $\rho$. The first term on the right side is called the \emph{direct} or \emph{Hartree term}. It is the classical electrostatic interaction energy of the density $\rho$ and it is the only nonlocal term in the LDA. The second term is the energy of the \emph{Uniform Electron Gas (UEG)}~\cite{DreGro-90,ParYan-94,GiuVig-05,LewLieSei-18}, containing all of the kinetic energy and the exchange-correlation energy in our convention. That is, $e_{\rm UEG}(\rho_0)$ is the ground state energy per unit volume of the infinite electron gas with the prescribed constant density $\rho_0$ over the whole space (from which the direct term has been dropped).
The rationale for the approximation~\eqref{eq:LDA} is to assume that the density is almost constant locally (in little boxes of volume $dx$), and to replace the local energy per unit volume by that of the infinite gas at that density $\rho(x)$.\footnote{It is often more convenient to fix the densities $\rho^\uparrow(x)$ and $\rho^\downarrow(x)$ of, respectively, spin-up and spin-down electrons instead of the total density $\rho(x)=\rho^\uparrow(x)+\rho^\downarrow(x)$. All our results apply similarly to this situation, as explained below in Remark~\ref{rmk:LSDA}.}

Our goal in this paper is to justify the approximation~\eqref{eq:LDA} in the appropriate regime where $\rho$ is flat in sufficiently large regions of $\R^3$. We will prove the following quantitative estimate
\begin{multline}
 \left|F_{\rm LL}(\rho)- \frac{1}{2}\int_{\R^3}\int_{\R^3}\frac{\rho(x)\rho(y)}{|x-y|}\,dx\,dy-\int_{\R^3}e_{\rm UEG}\big(\rho(x)\big)\,dx\right|\\
 \leq \eps \int_{\R^3}\big(\rho(x)+\rho(x)^2\big)\,dx
 +\frac{C(1+\eps)}{\eps}\int_{\R^3}|\nabla\sqrt\rho(x)|^2\,dx\\+\frac{C}{\eps^{4p-1}}\int_{\R^3}|\nabla\rho^\theta(x)|^p\,dx
 \label{eq:LDA_main_estim_intro}
\end{multline}
for all $\eps>0$, where $F_{\rm LL}(\rho)$ is the grand canonical version of the Levy-Lieb functional. The parameters $p>3$ and $0<\theta<1$ should satisfy some conditions which will be explained below. For instance, $p=4$ and $\theta=1/2$ is allowed. After optimizing over $\eps$, this justifies the LDA when the two gradient terms are much smaller than the local term
$$\begin{cases}
\dps \int_{\R^3}|\nabla\sqrt\rho(x)|^2\,dx\ll \int_{\R^3}\big(\rho(x)+\rho(x)^2\big)\,dx,\\[0.4cm]
\dps \int_{\R^3}|\nabla\rho^\theta(x)|^p\,dx\ll \int_{\R^3}\big(\rho(x)+\rho(x)^2\big)\,dx.   
  \end{cases}$$
For instance for a rescaled density in the form
$$\rho_N(x):=\rho\big(N^{-1/3}x\big)$$
with $\int_{\R^3}\rho=1$, we obtain after taking $\eps=N^{-1/12}$
\begin{equation*}
 \left|F_{\rm LL}(\rho_N)- \frac{N^{\frac53}}{2}\int_{\R^3}\int_{\R^3}\frac{\rho(x)\rho(y)}{|x-y|}\,dx\,dy-N\int_{\R^3}e_{\rm UEG}\big(\rho(x)\big)\,dx\right|
 \leq C N^{\frac{11}{12}}.
 \label{eq:LDA_main_rescaled_density_intro}
\end{equation*}

The bound~\eqref{eq:LDA_main_estim_intro} is, to our knowledge, the first estimate of this kind on the fundamental functional $F_{\rm LL}$. Although it should be possible to extract a definite value of the constant $C$ from our proof, it is probably very large and we have not tried to do it. The factor $1/\eps^{4p-1}$ is also quite large and it is an open problem to improve it. We hope that our work will stimulate more results on the functional $F_{\rm LL}$ in the regime of slowly varying densities. 

In physics and chemistry, the \emph{exchange-correlation energy} is defined by subtracting a kinetic energy term $T(\rho)$ from $F_{\rm LL}(\rho)$. In this paper we also derive a bound on $T(\rho)$ which, when combined with~\eqref{eq:LDA_main_estim_intro}, provides a bound on the exchange-correlation energy similar to~\eqref{eq:LDA_main_estim_intro}. This is explained below in Remark~\ref{rmk:exchange_correlation}.

\medskip

In the next section we provide the precise mathematical definition of $F_{\rm LL}$ and $e_{\rm UEG}$, and we state our main theorem containing the estimate~\eqref{eq:LDA_main_estim_intro}. In Section~\ref{sec:a_priori} we review some known \emph{a priori} estimates on $F_{\rm LL}$ and prove a new upper bound on the kinetic energy. Section~\ref{sec:proof} contains the proof of our main results. Finally, in Appendix~\ref{app:classical} we discuss a similar bound in the classical case where the kinetic energy is dropped, extending thereby our previous result in~\cite{LewLieSei-18}.

\subsubsection*{Acknowledgments.} The authors thank the Institut Henri Poincar\'e for its hospitality. This project has received funding from the European Research Council (ERC) under the European Union's Horizon 2020 research and innovation programme (grant agreements AQUAMS No 694227 of R.S. and MDFT No 725528 of M.L.).

\section{Main result}
 
 \subsection{The grand-canonical Levy-Lieb functional}
 
Let us consider a density $\rho\in L^1(\R^3,\R_+)$ such that $\sqrt\rho\in H^1(\R^3)$. Naturally we should assume in addition that $\int_{\R^3}\rho=N$ is an integer, but here we will work in the grand canonical ensemble where this is not needed.
The \emph{grand-canonical Levy-Lieb functional}~\cite{Levy-79,Lieb-83b,LewLieSei-18} is defined by
\begin{multline}
F_{\rm LL}(\rho):=\\\inf_{\substack{\Gamma_n=\Gamma_n^*\geq0\\ \sum_{n=0}^\ii\tr(\Gamma_n)=1\\ \sum_{n=1}^\ii \rho_{\Gamma_n}=\rho}} \Bigg\{\sum_{n=1}^\ii \tr_{\gH^n}\Bigg(-\sum_{j=1}^n\Delta_{x_j}+\sum_{1\leq j<k\leq n}\frac{1}{|x_j-x_k|}\Bigg)\Gamma_n\Bigg\}.
\label{eq:Levy-Lieb-grand-canonical}
\end{multline}
Here 
$$\gH^n:=L^2_a((\R^3\times \{1,...,q\})^n,\C)$$
is the $n$-particle space of antisymmetric square-integrable functions on $(\R^3\times \{1,...,q\})^n$, with $q$ spin states (for electrons $q=2$). The family of operators $\Gamma=\{\Gamma_n\}_{n\geq0}$ forms a grand-canonical mixed quantum state, that is, a state over the fermionic Fock space (commuting with the particle number operator). The density of each $\Gamma_n$ is defined by
\begin{multline*}
\rho_{\Gamma_n}(x)=\\
n\!\!\!\sum_{\substack{\sigma_1,...,\sigma_n\\ \in\{1,...,q\}}}\int_{\R^{3(n-1)}}
\Gamma_n(x,\sigma_1,x_2,...,x_n,\sigma_n;x,\sigma_1,x_2,...,x_n,\sigma_n)dx_2\cdots dx_n 
\end{multline*}
where $\Gamma_n(x_1,...,\sigma_n;x_1',...,\sigma'_n)$ is the kernel of the trace-class operator $\Gamma_n$. This kernel is such that 
\begin{align*}
&\Gamma_n(x_{\tau(1)},\sigma_{\tau(1)},...,x_{\tau(N)},\sigma_{\tau(N)}\;;\;x'_1,\sigma'_1,...,x'_N,\sigma'_N)\\
&\qquad=\Gamma_n(x_1,\sigma_1,...,x_N,\sigma_N\;;\;x'_{\tau(1)},\sigma'_{\tau(1)}...,x'_{\tau(N)},\sigma'_{\tau(N)})\\
&\qquad =\eps(\tau)\;\Gamma_n(x_1,\sigma_1,...,x_N,\sigma_N\;;\;x'_1,\sigma'_1...,x'_N,\sigma'_N)
\end{align*}
for every permutation $\tau\in\gS^N$ with signature $\eps(\tau)\in\{\pm1\}$.

If $\int_{\R^3}\rho=N\in\mathbb{N}$ and we restrict ourselves to mixed states $\Gamma$ where only $\Gamma_N$ is non-zero, we obtain Lieb's functional~\cite{Lieb-83b}. If we further assume that $\Gamma_N=|\Psi\rangle\langle\Psi|$ is a rank-one projection, then we find the original Levy or Hohenberg-Kohn functional~\cite{HohKoh-64,Levy-79}. It is well known~\cite{Lieb-83b} that working with mixed states has several advantages, in particular we obtain a convex function of $\rho$. 

The grand-canonical version~\eqref{eq:Levy-Lieb-grand-canonical} is less popular but still important physically.\footnote{The grand canonical functional $F_{\rm LL}$ is the weak-$\ast$ lower semi-continuous closure of the canonical Lieb functional~\cite{Lewin-11}. Hence it appears naturally in situations where some particles can be lost, e.g. in scattering processes.} It is also a convex function of $\rho$. The fact that we can appeal to states with an arbitrary number $n$ of particles (but still a fixed average number $\int_{\R^3}\rho$) will considerably simplify several technical parts of our study. We expect that our main result (Theorem~\ref{thm:LDA} below) holds the same for the canonical functionals, for which the energy is minimized over mixed or pure states with $N$ particles.

It is useful to subtract the direct term from $F_{\rm LL}$, hence to consider the energy
\begin{equation}
\boxed{E(\rho):=F_{\rm LL}(\rho)-\frac12\int_{\R^3}\int_{\R^3}\frac{\rho(x)\rho(y)}{|x-y|}dx\,dy.}
\label{eq:def_E_quantum}
\end{equation}
The (grand-canonical) \emph{exchange-correlation energy} is defined by 
$E_{\rm xc}(\rho):=E(\rho)-T(\rho)$
where 
\begin{equation}
T(\rho):=\inf_{\substack{\Gamma_n=\Gamma_n^*\geq0\\ \sum_{n=0}^\ii\tr(\Gamma_n)=1\\ \sum_{n=1}^\ii \rho_{\Gamma_n}=\rho}} \Bigg\{\sum_{n=1}^\ii \tr_{\gH^n}\Bigg(-\sum_{j=1}^n\Delta_{x_j}\Bigg)\Gamma_n\Bigg\}=\inf_{\substack{0\leq \gamma=\gamma^*\leq 1\\ \rho_\gamma=\rho}}\tr(-\Delta)\gamma
\label{eq:def_T}
\end{equation}
is the lowest possible kinetic energy. We will study the functional $T(\rho)$ in Section~\ref{sec:kinetic} below. 

\subsection{The Uniform Electron Gas}
In~\cite[Section 5]{LewLieSei-18}, we have defined the uniform electron gas, which is obtained in the limit when $\rho$ approaches a constant function in the whole space. This is believed to be the same as the ground state energy of Jellium, where the density is not necessarily constant but the electrons instead evolve in a constant background~\cite{LieNar-75}. This has recently been proved in the classical case in~\cite{LewLieSei-19b,CotPet-19b} and the same is expected in the quantum case.

The following result is a slight improvement of~\cite[Thm.~5.1]{LewLieSei-18}.

\begin{theorem}[Quantum Uniform Electron Gas]\label{thm:UEG_thermo_limit_quantum}
Let $\rho_0>0$. Let $\{\Omega_N\}\subset\R^3$ be a sequence of bounded connected domains with $|\Omega_N|\to\ii$, such that $\Omega_N$ has a uniformly regular boundary in the sense that
$$|\partial\Omega_N+B_r|\leq Cr|\Omega_N|^{2/3},\qquad\text{for all $r\leq |\Omega_N|^{1/3}/C$,}$$
for some constant $C>0$. Let $\delta_N>0$ be any sequence such that $\delta_N/|\Omega_N|^{1/3}\to0$ and $\delta_N|\Omega_N|^{1/3}\to\ii$. Let $\chi\in L^1(\R^3)$ be a radial non-negative function of compact support such that $\int_{\R^3}\chi=1$ and  $\int_{\R^3}|\nabla\sqrt\chi|^2<\ii$. Denote $\chi_\delta(x)=\delta^{-3}\chi(x/\delta)$. Then the following thermodynamic limit exists
\begin{equation}
\boxed{ \lim_{N\to\ii}\frac{E \big(\rho_0\1_{\Omega_N}\ast\chi_{\delta_N}\big)}{|\Omega_N|}
=e_{\rm UEG}\big(\rho_0\big) }
 \label{eq:thermodynamic_limit_quantum}
\end{equation}
where the function $e_{\rm UEG}$ is independent of the sequence $\{\Omega_N\}$, of $\delta_N$ and of~$\chi$. 
\end{theorem}

For more properties of the UEG energy $e_{\rm UEG}$ we refer to ~\cite{LewLieSei-18} and the references therein. In~\cite[Thm.~5.1]{LewLieSei-18} we rather optimized over the values of $\rho$ in the transition region around $\partial\Omega_N$. We were able to prove the simple limit~\eqref{eq:thermodynamic_limit_quantum} only when $\Omega_N$ is a tetrahedron. Using an upper bound on $E(\rho)$ that will be derived later in Proposition~\ref{prop:upper_bound}, we are now able to treat more reasonable limits in the form of~\eqref{eq:thermodynamic_limit_quantum}.  The proof of Theorem~\ref{thm:UEG_thermo_limit_quantum} is provided in Section~\ref{sec:proof_thermo_limit} below.

The function $\chi_{\delta_N}$ is used to regularize the function $\rho_0\1_{\Omega_N}$ which cannot be the density of a quantum state, since its square root is not in $H^1(\R^3)$~\cite{Lieb-83b}. The first condition $\delta_N/|\Omega_N|^{1/3}\to0$ implies that the smearing happens in a neighborhood of the boundary $\partial\Omega_N$ which has a negligible volume compared to $|\Omega_N|$. The second condition $\delta_N|\Omega_N|^{1/3}\to\ii$ ensures that the kinetic energy in the transition region stays negligible in the thermodynamic limit. 

\begin{remark}\rm 
The same result holds under the weaker condition that $\Omega_N$ has an $\eta$--regular boundary, which means that $|\partial\Omega_N+B_r|\leq C|\Omega_N|\eta\left(r|\Omega_N|^{-1/3}\right)$ for all $r\leq |\Omega_N|^{1/3}/C$, with $\eta(t)\to0$ when $t\to0^+$. The condition $\delta_N|\Omega_N|^{1/3}\to\ii$ is then replaced by $\delta_N^{-2}\eta(\delta_N|\Omega_N|^{-1/3})\to0$. 
\end{remark}

\subsection{The Local Density Approximation}

We are now able to state our main result. 

\begin{theorem}[Local Density Approximation]\label{thm:LDA}
Let $p>3$ and $0<\theta<1$ such that 
\begin{equation}
 2\leq p\theta\leq 1+\frac{p}{2}.
 \label{eq:hyp_p_theta}
\end{equation}
There exists a constant $C=C(p,\theta,q)$ such that 
\begin{multline}
 \left|E(\rho)- \int_{\R^3}e_{\rm UEG}\big(\rho(x)\big)\,dx\right|
 \leq \eps \int_{\R^3}\big(\rho(x)+\rho(x)^2\big)\,dx\\
 +\frac{C(1+\eps)}{\eps}\int_{\R^3}|\nabla\sqrt\rho(x)|^2\,dx+\frac{C}{\eps^{4p-1}}\int_{\R^3}|\nabla\rho^\theta(x)|^p\,dx
 \label{eq:LDA_main_estim}
\end{multline}
for every $\eps>0$ and every non-negative density $\rho\in L^1(\R^3)\cap L^2(\R^3)$ such that $\nabla \sqrt{\rho}\in L^2(\R^3)$ and $\nabla\rho^\theta\in L^p(\R^3)$.
\end{theorem}

The constant $C=C(p,\theta,q)$ in our estimate~\eqref{eq:LDA_main_estim} depends on the number of spin states $q$ ($q=2$ for electrons), in addition to the parameters $p$ and $\theta$. It diverges when $p\to 3^+$. If $p\to 3^+$ then we can take $\theta\to 5/6^-$. Our estimate therefore applies to densities $\rho$ with compact support, which vanish at the boundary of their support like $\delta(x)^a$ with $a>4/5$, where $\delta(x)={\rm d}(x,\partial\rho^{-1}(\{0\}))$. In particular, densities which vanish linearly are allowed. Our proof allows one to consider more singular densities, that is, to relax the constraint that $\theta p\leq 1+p/2$, but then the power of $\eps$ deteriorates. 

Our estimate~\eqref{eq:LDA_main_estim} is certainly not optimal and it is an interesting challenge to improve it. We conjecture that a similar inequality holds with $\rho+\rho^2$ replaced by $\rho^{4/3}+\rho^{5/3}$ which have the scaling of the Coulomb and kinetic energies, respectively. The higher power $\rho^2$ arises from the trial state used in our upper bound (Proposition~\ref{prop:upper_bound}) and it is used to control some errors appearing when merging quantum systems with overlapping supports. This is explained in Section~\ref{sec:upper_bound} below. Finally, the last gradient term in~\eqref{eq:LDA_main_estim} is used to control local variations of $\rho$ in $L^\ii$. One could expect gradient errors involving only $\int_{\R^3}|\nabla\sqrt\rho|^2$ and  $\int_{\R^3}|\nabla\rho^{1/3}|^2$ which are believed to arise in the gradient expansion of the uniform electron gas for, respectively, the Coulomb and kinetic energies.

One interesting case is when the density is given by a fixed function $\rho$ with $\int_{\R^3}\rho=1$, which is rescaled in the manner
$$\rho_N(x)=\rho(N^{-1/3}x).$$
After taking $\eps=N^{-1/12}$ in~\eqref{eq:LDA_main_estim} we obtain the following simple bound
\begin{multline}
 \left|E(\rho_N)- N\int_{\R^3}e_{\rm UEG}\big(\rho(x)\big)\,dx\right|\\
 \leq CN^{\frac{5}{12}}\int_{\R^3}|\nabla\sqrt\rho(x)|^2\,dx
+C N^{\frac{11}{12}}\int_{\R^3}\left(\rho(x)+\rho(x)^2+|\nabla\rho^\theta(x)|^p\right)\,dx.
 \label{eq:LDA_main_rescaled_density}
\end{multline}
It is conjectured~\cite{Kirzhnits-57,LanPer-80,LanMeh-83,LevPer-93,EngDre-11} that the next order in the expansion of $E(\rho_N)$ should involve the gradient correction to the kinetic energy
$$\int_{\R^3}\left|\nabla\sqrt{\rho_N}(x)\right|^2\,dx=N^{\frac13}\int_{\R^3}\left|\nabla\sqrt\rho(x)\right|^2\,dx$$
and the gradient correction to the Coulomb energy
$$\int_{\R^3}\left|\nabla\rho_N^{1/3}(x)\right|^2\,dx=N^{\frac13}\int_{\R^3}\left|\nabla\rho^{1/3}(x)\right|^2\,dx.$$
In particular, the next order should be proportional to $N^{1/3}$. It remains an open problem to establish this rigorously.

In the classical case where the kinetic energy is neglected, the limit of $E_{\rm cl}(\rho_N)/N$ was found in our previous work~\cite{LewLieSei-18}, but without a quantitative estimate on the remainder. We can give an estimate similar to~\eqref{eq:LDA_main_estim} in the classical case, with a lower power of $\eps$ in front of the gradient term. This is just a slight adaptation of the proof in~\cite{LewLieSei-18}, which is much easier than the quantum case. The argument is explained for completeness in Appendix~\ref{app:classical}. 

In the classical case the limit for $E_{\rm cl}(\rho_N)/N$ was later extended to Riesz interactions $|x|^{-s}$ and other dimensions $d\geq1$ in~\cite{CotPet-19}.  Although our result~\eqref{eq:LDA_main_estim} in the quantum case can probably be extended to other Riesz interactions by using ideas from~\cite{Fefferman-85,Hugues-85,Gregg-89,CotPet-19}, we only consider here the physically relevant 3D Coulomb case for shortness.

\begin{remark}[Canonical case]\rm 
We expect an inequality similar to~\eqref{eq:LDA_main_estim} for the (mixed) canonical version of $E(\rho)$ where $\int_{\R^3}\rho=N\in\mathbb{N}$ and $\Gamma_n=0$ for $n\neq N$. However our proof does not adapt in an obvious way to this case. 
\end{remark}

\begin{remark}[Exchange-correlation energy]\label{rmk:exchange_correlation}\rm 
In physics and chemistry, the LDA is usually expressed in terms of the \emph{exchange-correlation energy}. In the grand-canonical setting it is defined by
$$E_{\rm xc}(\rho)=E(\rho)-T(\rho),$$
with $T(\rho)$ the lowest possible kinetic energy~\eqref{eq:def_T}. The functional $T(\rho)$ is studied in Section~\ref{sec:a_priori} below, where it is proved that 
\begin{equation}
\left|T(\rho)-q^{-\frac23}c_{\rm TF}\int_{\R^3}\rho(x)^{\frac53}\,dx\right|
\leq \eps q^{-\frac23}\int_{\R^3}\rho(x)^{\frac53}\,dx+\frac{C}{\eps^{\frac{13}3}}\int_{\R^3}|\nabla\sqrt\rho(x)|^2\,dx
\label{eq:estim_T}
\end{equation}
with $c_{\rm TF}=3^{5/3}4^{1/3}\pi^{4/3}/5$ the Thomas-Fermi constant. The lower bound was derived by Nam~\cite{Nam-18} and we prove the missing upper bound (with a better power of $\eps$) in Theorem~\ref{thm:upper_bound_T} below. Actually, by following our proof of Theorem~\ref{thm:LDA} (simply discarding the Coulomb interaction) we can also prove a lower bound on $T(\rho)$, with an error similar to the right side of~\eqref{eq:LDA_main_estim} but with a smaller power of $\eps$ in front of $|\nabla\rho^\theta|^p$. This provides the following estimate on the exchange-correlation energy
\begin{multline}
 \left|E_{\rm xc}(\rho)- \int_{\R^3}e_{\rm UEG}\big(\rho(x)\big)\,dx+q^{-\frac23}c_{\rm TF}\int_{\R^3}\rho(x)^{\frac53}\,dx\right|\\
 \leq \eps \int_{\R^3}\big(\rho(x)+\rho(x)^2\big)\,dx
 +\frac{C(1+\eps)}{\eps}\int_{\R^3}|\nabla\sqrt\rho(x)|^2\,dx\\+\frac{C}{\eps^{4p-1}}\int_{\R^3}|\nabla\rho^\theta(x)|^p\,dx.
 \label{eq:LDA_exchange_correlation}
\end{multline}
For a rescaled density $\rho_N(x)=\rho(x/N^{1/3})$ we obtain the same rate of convergence $N^{11/12}$ as in~\eqref{eq:LDA_main_rescaled_density}. 
\end{remark}

\begin{remark}[Local Spin Density Approximation]\label{rmk:LSDA}\rm 
In practice, it is often convenient to not fix the total density but, rather, the density of each spin component
\begin{multline*}
\rho_\sigma(x)=\\
\sum_{n\geq1}n \sum_{\substack{\sigma_2,...,\sigma_n\\ \in\{1,...,q\}}}\int_{\R^{3(n-1)}}
\Gamma_n(x,\sigma,x_2,...,x_n,\sigma_n;x,\sigma,x_2,...,x_n,\sigma_n)dx_2\cdots dx_n 
\end{multline*}
for $\sigma\in\{1,...,q\}$. Similarly as in Theorem~\ref{thm:UEG_thermo_limit_quantum} one can define the corresponding spin-polarized UEG energy $e_{\rm UEG}(\rho_1,...,\rho_q)$ of the uniform electron gas where the electrons of spin $\sigma$ are assumed to have the constant density $\rho_\sigma$. By following the arguments in this paper, one can then prove the estimate similar to~\eqref{eq:LDA_main_estim}
\begin{multline}
 \left|E(\rho_1,...,\rho_q)- \int_{\R^3}e_{\rm UEG}\big(\rho_1(x),...,\rho_q(x)\big)\,dx\right|
 \leq \eps \int_{\R^3}\big(\rho(x)+\rho(x)^2\big)\,dx\\
 +\frac{C(1+\eps)}{\eps}\int_{\R^3}|\nabla\sqrt\rho(x)|^2\,dx+\frac{C}{\eps^{4p-1}}\int_{\R^3}|\nabla\rho^\theta(x)|^p\,dx.
 \label{eq:LDA_main_estim_with_spin}
\end{multline}
It is only for simplicity of notation that we work with the total density $\rho=\sum_{\sigma=1}^q\rho_\sigma$. 
\end{remark}

\section{A priori estimates on $T(\rho)$ and $E(\rho)$}\label{sec:a_priori}

Lower bounds on $E(\rho)$ in~\eqref{eq:def_E_quantum} are well known and will be recalled below. Upper bounds are somewhat difficult to derive due to the constraint that the quantum states considered need to have the exact given density $\rho$. In this section we prove an upper bound on the best kinetic energy and use it to derive an upper bound on $E(\rho)$. Because our bounds are of independent interest we work in this section in any dimension $d\geq1$. First we quickly recall the known lower bounds.

\subsection{Known lower bounds}

We recall that the Lieb-Thirring inequality~\cite{LieThi-75,LieThi-76,LieSei-09} states that there exists a positive constant $c_{\rm LT}=c_{\rm LT}(d)>0$ such that 
\begin{equation}
\tr(-\Delta)\gamma\geq q^{-\frac2d}c_{\rm LT}\int_{\R^d} \rho_\gamma(x)^{1+\frac2d}\,dx
\label{eq:Lieb-Thirring}
\end{equation}
for every self-adjoint operator $\gamma$ on $L^2(\R^d,\C^q)$ such that $0\leq\gamma\leq1$. The best constant $c_{\rm LT}$ is unknown but has been conjectured to be the semi-classical constant 
\begin{equation}
c_{\rm TF}=\frac{ 4\pi^2 d}{(d+2)} \left( \dfrac{d}{| \mathbb{S}^{d-1} |} \right)^{\frac2d}
 \label{eq:c_TF}
\end{equation}
in dimension $d\geq3$. In~\cite{Nam-18}, Nam has proved that 
\begin{equation}
\tr(-\Delta)\gamma\geq q^{-\frac2d}c_{\rm TF}\left(1-\eps\right)\int_{\R^d}\rho_\gamma(x)^{1+\frac2d}dx
 -\frac{\kappa}{\eps^{3+\frac4d}}\int_{\R^d}|\nabla\sqrt{\rho_\gamma}(x)|^2\,dx
 \label{eq:lower_bound_Nam}
\end{equation}
for every $\eps>0$ and some constant $\kappa=\kappa(d)$, in all space dimensions $d\geq1$.

We also recall the Hoffmann-Ostenhof inequality~\cite{Hof-77} which states that 
\begin{equation}
\tr(-\Delta)\gamma\geq \int_{\R^d} \left|\nabla\sqrt{\rho_\gamma}(x)\right|^2\,dx
\label{eq:Hoffmann-Ostenhof}
\end{equation}
and always imposes that $\sqrt{\rho}\in H^1(\R^d)$. The inequality~\eqref{eq:Hoffmann-Ostenhof} does not require the fermionic constraint $0\leq\gamma\leq1$.

The Lieb-Oxford inequality~\cite{Lieb-79,LieOxf-80,ChaHan-99,LieSei-09} states that the total Coulomb energy is bounded from below by
\begin{multline}
\sum_{n=1}^\ii \tr_{\gH^n}\Bigg(\sum_{1\leq j<k\leq n}\frac{1}{|x_j-x_k|}\Bigg)\Gamma_n\Bigg\}\\
\geq \frac12\int_{\R^3}\int_{\R^3}\frac{\rho(x)\rho(y)}{|x-y|}dx\,dy-1.64\int_{\R^3}\rho(x)^{\frac43}\,dx
\label{eq:Lieb-Oxford}
\end{multline}
where $\Gamma=\{\Gamma_n\}$ is a grand-canonical quantum state satisfying the conditions in~\eqref{eq:def_E_quantum}. Inspired by~\cite{BenBleLos-12}, this bound was recently generalized in~\cite{LewLie-15} to
\begin{multline}
\sum_{n=1}^\ii \tr_{\gH^n}\Bigg(\sum_{1\leq j<k\leq n}\frac{1}{|x_j-x_k|}\Bigg)\Gamma_n\Bigg\}\\
\geq \frac12\int_{\R^3}\int_{\R^3}\frac{\rho(x)\rho(y)}{|x-y|}dx\,dy-\left(\frac35 \left(\frac{9\pi}{2}\right)^{\frac13}+\eps\right)\int_{\R^3}\rho(x)^{\frac43}\,dx\\-\frac{0.001206}{\eps^3}\int_{\R^3}|\nabla\rho(x)|\,dx
\label{eq:Lieb-Oxford-gradient}
\end{multline}
but the constant $(3/5) ({9\pi}/{2})^{1/3}\simeq 1.4508$ is not expected to be the optimal Lieb-Oxford constant. 

Using~\eqref{eq:Lieb-Thirring} together with~\eqref{eq:Lieb-Oxford}, we obtain the following.

\begin{corollary}[Lower bound on $E(\rho)$]\label{cor:lower_bound_E_rho}
We have 
\begin{equation}
E(\rho)\geq q^{-\frac23}c_{\rm LT}\int_{\R^3} \rho(x)^{\frac53}\,dx-1.64\int_{\R^3}\rho(x)^{\frac43}\,dx
\label{eq:lower_bound_E}
\end{equation}
for every $\rho\geq0$ such that $\sqrt{\rho}\in H^1(\R^d)$.
\end{corollary}

The constants can be improved at the expense of adding gradient corrections, using~\eqref{eq:lower_bound_Nam} and~\eqref{eq:Lieb-Oxford-gradient}.

\subsection{Upper bound on the best kinetic energy}\label{sec:kinetic}
Let us recall that the lowest possible kinetic energy of a fixed density $\rho\in L^1(\R^d,\R_+)$ with $\sqrt\rho\in H^1(\R^d)$ reads
$$T(\rho):=\min_{\substack{0\leq \gamma=\gamma^*\leq 1\\ \rho_\gamma=\rho}}\tr(-\Delta)\gamma.$$
In~\cite{LewLieSei-18}, we have shown that for $\rho\leq1$,
$$T(\rho)\leq \int_{\R^d}|\nabla\sqrt\rho|^2+q^{-\frac2d}c_{\rm TF}\int_{\R^d}\rho$$
by using the trial state
$$\gamma=\sqrt{\rho(x)}\,\1\left(-\Delta\leq \frac{d+2}d c_{\rm TF}q^{-\frac2d}\right)\sqrt{\rho(x)}.$$
Our goal in this section is to prove a similar bound without the assumption that $\rho\leq1$. Coherent states~\cite{Lieb-81b} can usually give good bounds on the kinetic energy but they do not preserve the density. The main difficulty here is to construct a state having the exact given density $\rho$. The next result says that the semi-classical approximation to the kinetic energy is an upper bound to the exact $T(\rho)$, up to some gradient corrections. 

\begin{theorem}[Upper bound on the best kinetic energy]\label{thm:upper_bound_T}
There are two constants $\kappa_1,\kappa_2>0$ depending only on the space dimension $d$ such that 
\begin{equation}
 T(\rho)\leq q^{-\frac2d}c_{\rm TF}\left(1+\kappa_1\eps\right)\int_{\R^d}\rho(x)^{1+\frac2d}dx
 +\frac{\kappa_2(1+\sqrt\eps)^2}{\eps}\int_{\R^d}|\nabla\sqrt{\rho}(x)|^2\,dx,
 \label{eq:upper_bound_T_kappa}
\end{equation}
for every $\rho\in L^1(\R^d,\R_+)$ with $\sqrt{\rho}\in H^1(\R^d)$ and every $\eps>0$, where $c_{\rm TF}$ is the Thomas-Fermi constant~\eqref{eq:c_TF}.
\end{theorem}

Note that the gradient correction in~\eqref{eq:upper_bound_T_kappa} has a better behavior in $\eps$ than in Nam's lower bound~\eqref{eq:lower_bound_Nam}. 

In dimension $d=1$, March and Young have given the proof of a better estimate without the parameter $\eps$, in~\cite[Eq. (9)]{MarYou-58}:
$$T(\rho)\leq q^{-2}c_{\rm TF}\int_\R \rho(x)^3\,dx+\int_\R\left|\big(\sqrt{\rho}\big)'(x)\right|^2\,dx.$$
In the same paper they also state a result in 3D (for a constant $c>c_{\rm TF}$) but the proof has a mistake. This was mentioned as a conjecture in~\cite[Sec. 5.B]{Lieb-83b}. Our result~\eqref{eq:upper_bound_T_kappa} can therefore be seen as a solution to the March-Young problem. We conjecture that a similar bound holds without the parameter $\eps$ in dimension $d=2,3$ as well. 

\begin{remark}[Explicit constants in 3D]\rm 
In dimension $d=3$ one can take
$$\kappa_1=1,\qquad \kappa_2=48 \qquad\text{in~\eqref{eq:upper_bound_T_kappa}.}$$
These constants are not optimal and they are only displayed for concreteness. Our proof allows to slightly improve the constants under the assumption that $\eps$ is small enough. For instance for 
$\eps\leq 1$, we have the better inequality
\begin{equation}
 T(\rho)\leq q^{-\frac23}c_{\rm TF}\left(1+\frac{\eps}{15}\right)\int_{\R^3}\rho(x)^{\frac53}dx
 +\frac{19}{\eps}\int_{\R^3}|\nabla\sqrt{\rho}(x)|^2\,dx.
 \label{eq:upper_bound_T_eps}
\end{equation}
\end{remark}

\begin{proof}[Proof of Theorem~\ref{thm:upper_bound_T}]
For simplicity we only write the proof in the no-spin case $q=1$. Recall that the free Fermi sea 
$$P_r=\1\left(-\Delta\leq \frac{d+2}d c_{\rm TF}r^{2/d}\right)$$ 
has the constant density $\rho_{P_r}=r$ and the constant kinetic energy density $c_{\rm TF}r^{1+2/d}$. In particular, if we take 
$$\gamma=\sqrt{f(x)}\,\1\left(-\Delta\leq \frac{d+2}dc_{\rm TF}r^{2/d}\right)\sqrt{f(x)}$$
for some $f\geq0$, we obtain
\begin{equation}
\rho_\gamma=rf(x),\qquad \tr(-\Delta)\gamma=r\int_{\R^d}|\nabla \sqrt{f}(x)|^2\,dx+c_{\rm TF}r^{1+\frac2d}\,\int_{\R^d}f(x)\,dx.
\label{eq:computation_simple_trial_state}
\end{equation}
See for instance~\cite[Sec.~5]{LewLieSei-18} for details.
In addition, we have in the sense of operators $0\leq\gamma\leq f(x)$, hence $\gamma$ is a fermionic one-particle density matrix under the additional condition that $f\leq1$. 

Let now $\eta\geq0$ be a smooth non-negative function such that 
\begin{equation}
\int_0^\ii\eta(t)\,dt=1,\qquad \int_0^\ii\eta(t)\,\frac{dt}{t}\leq1.
\label{eq:cond_eta}
\end{equation}
Using the smooth layer cake principle
$$\rho(x)=\int_0^\ii \eta\left(\frac{t}{\rho(x)}\right)\,dt$$
we introduce the trial state
$$\gamma=\int_0^\ii  \sqrt{\eta\left(\frac{t}{\rho(x)}\right)}\;\1\left(-\Delta\leq \frac{d+2}d c_{\rm TF}\,t^{2/d}\right)\;\sqrt{\eta\left(\frac{t}{\rho(x)}\right)}\;\frac{dt}{t}.$$
In the sense of operators, we have 
$$0\leq \gamma\leq \int_0^\ii  \eta\left(\frac{t}{\rho(x)}\right)\;\frac{dt}{t}=\int_0^\ii\eta(t)\,\frac{dt}{t}\leq1.$$
In addition, $\gamma$ has the required density
$$\rho_\gamma(x)=\int_0^\ii  \eta\left(\frac{t}{\rho(x)}\right)\,dt=\rho(x).$$
Hence $\gamma$ is admissible and it can be used to get an upper bound on $T(\rho)$. 
From~\eqref{eq:computation_simple_trial_state}, its kinetic energy is
\begin{multline*}
\tr(-\Delta)\gamma=\int_{\R^d}dx\int_0^\ii dt  \left|\nabla_x\sqrt{\eta\left(\frac{t}{\rho(x)}\right)}\right|^2\\
+c_{\rm TF}\int_{\R^d}\rho(x)^{1+\frac2d}dx\int_0^\ii \eta(t)t^{\frac2d}\,dt.
\end{multline*}
Note that
\begin{align*}
\int_{\R^d}dx\int_0^\ii dt  \left|\nabla_x\sqrt{\eta\left(\frac{t}{\rho(x)}\right)}\right|^2=& \int_{\R^d}dx\int_0^\ii dt  \frac{\left|\nabla_x\eta\left(\frac{t}{\rho(x)}\right)\right|^2}{4\eta\left(\frac{t}{\rho(x)}\right)}\\
=& \int_{\R^d}dx\frac{|\nabla\rho(x)|^2}{4\rho(x)^{4}}\int_0^\ii t^2dt  \frac{\left|\eta'\left(\frac{t}{\rho(x)}\right)\right|^2}{\eta\left(\frac{t}{\rho(x)}\right)}\\
=& \int_{\R^d}|\nabla\sqrt{\rho}(x)|^2\,dx\int_0^\ii \frac{t^2\eta'(t)^2}{\eta(t)}\,dt.
\end{align*}
Hence we have proved that 
\begin{multline*}
\tr(-\Delta)\gamma=\int_{\R^d}|\nabla\sqrt{\rho}(x)|^2\,dx\int_0^\ii \frac{t^2\eta'(t)^2}{\eta(t)}\,dt\\
+c_{\rm TF}\int_{\R^d}\rho(x)^{1+\frac2d}dx\int_0^\ii \eta(t)t^{\frac2d}\,dt
\end{multline*}
for all $\eta\geq0$ satisfying the two constraints~\eqref{eq:cond_eta}.
The smallest constant we can get in front of the $\rho^{1+2/d}$ term is $c_{\rm TF}$, by concentrating $\eta$ at the point $t=1$, but this makes the other term blow up. If we fix 
$$\int_0^\ii \eta(t)t^{\frac2d}\,dt=1+\eps$$
then the best constant we can get in front of the gradient term is given by the variational problem
$$C(\eps):=\inf_{\substack{\eta\geq0\\ \int_0^\ii\eta=1\\ \int_0^\ii\eta/t\leq 1\\ \int_0^\ii t^{2/d}\eta\leq 1+\eps}}\int_0^\ii \frac{t^2\eta'(t)^2}{\eta(t)}\,dt.$$
We claim that $C(\eps)\leq {\rm const.}(1+\eps^{-1})$ for $\eps$ small enough, which we prove by an appropriate choice of $\eta$. 

Let us first take, for instance,
\begin{equation}
\eta_\eps(t)=\frac{3}{2\eps^3}(t-1)^2\1(1\leq t\leq 1+\eps)+\frac{3}{2\eps^3}(1+2\eps-t)^2\1(1+\eps\leq t\leq 1+2\eps).
\label{eq:choice_eta_1}
\end{equation}
Then $\int \eta_\eps=1$ and $\int \eta_\eps(t)/t\,dt\leq1$ since $\eta_\eps$ is supported on $[1,\ii)$. 
Using the simple bounds
$$\int_0^\ii \eta(t)t^{\frac2d}\,dt\leq (1+2\eps)^{2/d}$$
and 
$$\int_0^\ii \frac{t^2\eta'(t)^2}{\eta(t)}\,dt\leq (1+2\eps)^2\int_0^\ii \frac{\eta'(t)^2}{\eta(t)}\,dt=\frac{12(1+2\eps)^2}{\eps^2},$$
we obtain (after changing $\eps$ into $\eps/2$) 
\begin{equation}
 T(\rho)\leq c_{\rm TF}(1+\eps)^{\frac2d}\int_{\R^d}\rho(x)^{1+\frac2d}dx
 +\frac{48(1+\eps)^2}{\eps^2}\int_{\R^d}|\nabla\sqrt{\rho}(x)|^2\,dx,\qquad \forall \eps>0.
 \label{eq:upper_bound_T}
\end{equation}

The behavior of the correction in front of the semi-classical term $c_{\rm TF}$ is not optimal for small $\eps$. It can be replaced by $1+\kappa_1\eps^2$, for $\eps\leq1$. To see this we slightly translate the function~\eqref{eq:choice_eta_1} to the left by an amount $-\eps b$ and introduce 
\begin{multline}
\eta_{\eps,b}(t)=\frac{3}{2\eps^3}(t-1+\eps b)^2\1(1-\eps b\leq t\leq 1+(1-b)\eps)\\+\frac{3}{2\eps^3}(1+(2-b)\eps-t)^2\1(1-\eps b+\eps\leq t\leq 1+(2-b)\eps).
\label{eq:choice_eta_2}
\end{multline}
Then we have
\begin{multline*}
\int_0^\ii\eta_{\eps,b}(t)\frac{dt}{t}=1 + (b-1 )\eps  + \left(\frac{11}{10} - 2 b + b^2\right) \eps^2\\
+ \left(-\frac{13}{10} + \frac{33}{10}\eps - 3 b^2 + b^3\right) \eps^3+O(\eps^4) 
\end{multline*}
and
\begin{multline*}
\int_0^\ii\eta_{\eps,b}(t)t^{\frac2d}\,dt=1 - \frac2d (b-1) \eps\\ + \frac{1}{10d^2}\big(22 - 40 b + 20 b^2 - 11 d + 20 b d - 10 b^2 d\big) \eps^2+O(\eps^3). 
\end{multline*}
The unique $b_\eps$ such that $\int_0^\ii\eta_{\eps,b_\eps}(t)\frac{dt}{t}=1$ satisfies
$$b_\eps=1-\frac{\eps}{10}-\frac{3}{350}\eps^3+O(\eps^4)$$
and for this $b_\eps$ we have
$$\int_0^\ii\eta_{\eps,b_\eps}(t)t^{\frac2d}\,dt=1 +\frac{2+d}{10d^2}\eps^2+O(\eps^4),$$
$$\int_0^\ii \frac{t^2\eta_{\eps,b_\eps}'(t)^2}{\eta_{\eps,b_\eps}(t)}\,dt=\frac{12}{\eps^2}+O(1).$$
This is how we can get~\eqref{eq:upper_bound_T_eps} for $\eps$ small enough (after replacing $\eps^2$ by $\eps$). 
\end{proof}

\begin{remark}\rm 
In the 3D case we can take for instance $b=1-\eps/10-4\eps^3/350$. One can then verify that 
$$\int_0^\ii\eta_{\eps,b}(t)\frac{dt}{t}\leq 1,\quad \int_0^\ii\eta_{\eps,b}(t)t^{\frac23}\,dt\leq 1+\frac{\eps^2}{15}$$
and
$$\int_0^\ii \frac{t^2\eta_{\eps,b_\eps}'(t)^2}{\eta_{\eps,b_\eps}(t)}\,dt\leq\frac{19}{\eps^2}$$
for all $\eps\leq1$. Hence
\begin{equation}
 T(\rho)\leq c_{\rm TF}\left(1+\frac{\eps^2}{15}\right)\int_{\R^3}\rho(x)^{\frac53}dx
 +\frac{19}{\eps^2}\int_{\R^3}|\nabla\sqrt{\rho}(x)|^2\,dx,\qquad\forall \eps\leq1.
 \label{eq:upper_bound_T_1}
\end{equation}
Combining with~\eqref{eq:upper_bound_T}, we find the estimate~\eqref{eq:upper_bound_T_kappa} for $\kappa_1=1$ and $\kappa_2=48$.
\end{remark}

\subsection{Upper bound on $E(\rho)$}
It is well known that any fermionic one-particle density matrix $\gamma$ (i.e., an operator satisfying $0\leq\gamma=\gamma^*\leq1$) is representable by a quasi-free state $\Gamma_\gamma$ in Fock space~\cite{BacLieSol-94}. The two-particle density matrix of such a state is given by Wick's formula
\begin{multline}
\Gamma_\gamma^{(2)}(x_1,\sigma_1,x_2,\sigma_2;y_1,\sigma'_1,y_2,\sigma'_2)\\
=\gamma(x_1,\sigma_1;y_1,\sigma'_1)\gamma(x_2,\sigma_2;y_2,\sigma_2')-\gamma(x_1,\sigma_1;x_2,\sigma_2)\gamma(y_1,\sigma'_1;y_2,\sigma'_2).
\label{eq:2PDM_quasi_free_state}
\end{multline}
In particular, the corresponding interaction energy with pair potential $w$ is
$$\frac12 \iint_{\R^d\times\R^d}w(x-y)\left(\rho(x)\rho(y)-\sum_{\sigma,\sigma'=1}^q|\gamma(x,\sigma;y,\sigma')|^2\right)\,dx\,dy.$$
From this we immediately obtain the following.

\begin{corollary}[Upper bound on $E(\rho)$ in dimension $d\geq1$]\label{cor:upper_bound_E_rho}
We have 
\begin{equation}
E(\rho)\leq q^{-\frac2d}c_{\rm TF}(1+\kappa_1\eps)\int_{\R^d} \rho(x)^{1+\frac2d}\,dx+\frac{\kappa_2(1+\sqrt\eps)^2}{\eps}\int_{\R^d}|\nabla\sqrt{\rho}(x)|^2\,dx
\label{eq:upper_bound_E}
\end{equation}
for every $\rho\geq0$ such that $\sqrt{\rho}\in H^1(\R^d)$.
\end{corollary}

This is for the grand-canonical version~\eqref{eq:def_E_quantum} of the Levy-Lieb functional which is the object of concern in this paper. It was proved in~\cite{Lieb-81} that any fermionic $\gamma$ with integer trace $N=\tr(\gamma)$ is also the one-particle density matrix of an $N$-particle mixed state $\Gamma$ on the fermionic space $\gH^N$, such that the corresponding two-particle density matrix satisfies
$$\Gamma^{(2)}\leq \Gamma^{(2)}_\gamma$$
in the sense of operators. From the positivity of the Coulomb potential we deduce immediately the following result for mixed canonical states.

\begin{corollary}[Upper bound in the mixed canonical case]
Let $\rho\in L^1(\R^d,\R_+)$ be such that $\int_{\R^d}\rho=N\in\mathbb{N}$, and $\sqrt{\rho}\in H^1(\R^d)$. Then there exists a mixed state $\Gamma$ on the fermionic space $\bigwedge_1^NL^2(\R^d\times\{1,...,q\})$ such that $\rho_\Gamma=\rho$ and
\begin{multline*}
\tr\left(\sum_{j=1}^N-\Delta_{x_j}+\sum_{1\leq j<k\leq N}\frac{1}{|x_j-x_k|}\right)\Gamma\\
\leq q^{-\frac2d}c_{\rm LT}(1+\kappa_1\eps)\int_{\R^d} \rho(x)^{1+\frac2d}\,dx+\frac{\kappa_2(1+\sqrt\eps)^2}{\eps}\int_{\R^d}|\nabla\sqrt{\rho}(x)|^2\,dx.
\end{multline*}
\end{corollary}

\section{Proof of Theorems~\ref{thm:UEG_thermo_limit_quantum} and~\ref{thm:LDA}}\label{sec:proof}

Our proof is divided into several steps. The first is to show that the energy is essentially local, that is, to prove that 
\begin{equation}
E(\rho)\approx \sum_k  E\big(\rho\chi_k\big)
\label{eq:energy_is_almost_local}
\end{equation}
where $\{\chi_k\}$ is a smooth partition of unity, $\sum_k\chi_k=1$. The precise statement of~\eqref{eq:energy_is_almost_local} will involve upper and lower bounds, as well as an average over the translations, rotations and dilations of the partition itself. The lower bound was indeed already shown in~\cite{LewLieSei-18} using the Graf-Schenker inequality~\cite{GraSch-95}. The upper bound is the main new ingredient of our proof. The two bounds are derived in Sections~\ref{sec:upper_bound} and~\ref{sec:lower_bound}. This will allow us to provide a rather simple proof of Theorem~\ref{thm:UEG_thermo_limit_quantum} in Section~\ref{sec:proof_thermo_limit}, using the convergence for tetrahedra which will be studied in Section~\ref{sec:CV_rate_tetrahedra}.

In Section~\ref{sec:deviation} we will estimate the deviation of the energy when we replace $\rho\chi_k$ by a constant function, say $\rho(x_k)\chi_k$ for some $x_k$ in the support of $\chi_k$. If $\rho$ is essentially constant in the corresponding region, the error will be small, but if $\rho$ is not constant we bound the energy using some gradient terms, utilizing our upper bound~\eqref{eq:upper_bound_E}.

After showing the Lipschitz regularity of $e_{\rm UEG}$ in Section~\ref{sec:Lipschitz_UEG}, we will be able to conclude the proof of Theorem~\ref{thm:LDA} in Section~\ref{sec:conclusion}. We replace $E(\rho(x_k)\chi_k)$ by $\int_{\R^3}e_{\rm UEG}(\rho(x))\chi_k(x)\,dx$ when $\rho$ is large enough on ${\rm supp}(\chi_k)$, using some quantitative estimates derived for tetrahedra in Section~\ref{sec:CV_rate_tetrahedra}.

In the rest of the paper we call $C$ a generic constant which can sometimes change from line to line, but which only depends on $q$ (the number of spin states) and $p,\theta$, the two parameters appearing in the statement of Theorem~\ref{thm:LDA}.

\subsection{Upper bound in terms of local densities}\label{sec:upper_bound}

Our main goal here is to give an upper bound on the energy $E(\rho)$ by splitting $\rho$ into a sum of local densities. 
In the classical case, we have the exact subadditivity property (see~\cite[Lem.~2.5]{LewLieSei-18})
$$E_{\rm cl}(\rho_1+\rho_2)\leq E_{\rm cl}(\rho_1)+E_{\rm cl}(\rho_2)$$ 
which considerably simplifies the analysis and was one of the main tools of our previous work~\cite{LewLieSei-18}. In particular we immediately find an upper bound in the form
$$E_{\rm cl}(\rho)\leq \sum_kE_{\rm cl}(\rho\chi_k)$$
for a partition of unity $\chi_k$. In the quantum case this is not as easy. The first difficulty is that we cannot cut sharply and have to use a smooth partition of unity. This has the consequence that neighboring local densities overlap. But then, for two densities $\rho_1$ and $\rho_2$ with overlapping support, it is not obvious how to relate $E(\rho_1+\rho_2)$ with $E(\rho_1)$ and $E(\rho_2)$. This is due to the fermionic nature of the electrons which puts a very strong constraint on trial states. If we take two trial quantum states for $\rho_1$ and $\rho_2$, we cannot simply take their tensor product and use it as a trial state for $\rho_1+\rho_2$. The tensor product does not have the fermionic symmetry, and if we anti-symmetrize it the density is not equal to $\rho_1+\rho_2$ anymore. 

For this reason, we will use an incomplete partition of unity with holes, in order to make sure that the local quantum states are not overlapping. Since we want to get the exact density, holes are however in principle not allowed. Instead of filling the holes with electrons, our idea is to rather average over all the possible rotations, translations and dilations of the partition of unity, which will make the holes disappear in average. All the arguments of this section apply the same to a tiling made of cubes, but for a better matching with the lower bound we will consider a tiling made of tetrahedra. Our lower bound relies on the Graf-Schenker inequality~\cite{GraSch-95} which requires the use of tetrahedra. 

Let us consider the unit cube $C_1=(-1/2,1/2)^3$, which is the union of 24 disjoint identical tetrahedra $\bDelta_1,...,\bDelta_{24}$, all of volume $1/24$. Since the cube can be repeated in the whole space, we obtain a tiling of $\R^3$ with tetrahedra:
\begin{equation}
 \R^3=\bigcup_{z\in\Z^3}\bigcup_{j=1}^{24}(\overline{\bDelta_j}+z)
 \label{eq:tiling}
\end{equation}
and the corresponding partition of unity
$$\sum_{z\in\Z^3}\sum_{j=1}^{24}\1_{\ell\bDelta_j}(x-\ell z)\equiv1,\qquad \text{for a.e. $x\in\R^3$},$$
for any fixed size $\ell>0$ of the tiles. Any $\bDelta_j$ can be written as $\bDelta_j=\mu_j\bDelta$ where $\bDelta$ is a reference tetrahedron with $0$ as its center of mass. Here $\mu_j=(z_j,R_j)\in C_1\times SO(3)$ is an appropriate translation and rotation, which acts as $\mu_j x=R_jx-z_j$, hence $\bDelta_j=R_j\bDelta-z_j$. In each tetrahedron we now place the regularized characteristic function
\begin{equation}
 \boxed{\chi_{\ell,\delta,j}:=\frac{1}{(1-\delta/\ell)^3}\1_{\ell\mu_j(1-\delta/\ell)\bDelta}\ast \eta_\delta.} 
 \label{eq:def_chi_ell_tiling}
\end{equation}
Here $\eta_\delta(x)=(10/\delta)^{3}\eta_1(10x/\delta)$ where $\eta_1$ is a fixed $C^\ii_c$ non-negative radial function with support in the unit ball and such that $\int_{\R^3}\eta_1=1$. Assuming that $\delta\leq\ell/2$, the function $\chi_{\ell,\delta,j}$ has its support well inside $\ell \bDelta_j$, at a distance proportional to $\delta$ from its boundary. The prefactor has been chosen to ensure that 
$$\int_{\R^3}\chi_{\ell,\delta,j}=\frac{\ell^3}{24}=|\ell\bDelta|.$$
The function
$$\sum_{z\in\Z^3}\sum_{j=1}^{24} \chi_{\ell,\delta,j}(x-\ell z)$$
is equal to $(1-\delta/\ell)^{-3}>1$ inside the tiles but vanishes in a neighborhood of the boundary of the tiles. It is the incomplete partition of unity which we have mentioned above.
We obtain a partition of unity after averaging over the translations of the tiling:
\begin{equation}
\frac{1}{\ell^3}\int_{C_\ell}\sum_{z\in\Z^3}\sum_{j=1}^{24} \chi_{\ell,\delta,j}(x-\ell z-\tau)\,d\tau=1,\qquad \text{for a.e. $x\in\R^3$}.
\label{eq:resolution_identity}
\end{equation}
Here $C_\ell=(-\ell/2,\ell/2)^3=\ell C_1$ is the cube of side length $\ell$. 
This is because for any $f\in L^1(\R^3)$,
\begin{align*}
\int_{C_\ell}\sum_{z\in\Z^3}f(x-\ell z-\tau)\,d\tau&=\int_{C_\ell}\sum_{z\in \ell\Z^3}f(x-z-\tau)\,d\tau\\
&=\int_{\R^3}f(x-\tau)\,d\tau=\int_{\R^3}f(\tau)\,d\tau. 
\end{align*}

The main result of this section is the following upper bound.

\begin{proposition}[Upper bound in terms of local densities]\label{prop:upper_bound}
There exists a universal constant $C$ such that for any $\sqrt{\rho}\in H^1(\R^3)$, any $0<\delta<\ell/2$ and any $0<\alpha<1/2$, 
\begin{multline}
 E(\rho)\leq
  \left(\int_{1-\alpha}^{1+\alpha}\frac{ds}{s^4}\right)^{-1}\int_{1-\alpha}^{1+\alpha}\frac{dt}{t^4}\int_{SO(3)}dR\int_{C_{t\ell}}\frac{d\tau}{(t\ell)^3}\times\\
  \times \sum_{z\in\Z^3}\sum_{j=1}^{24} E\Big(\chi_{t\ell,t\delta,j}(R\,\cdot\,-t\ell z-\tau)\rho\Big)
 +C\delta^2\log(\alpha^{-1})\int_{\R^3}\rho^2.
 \label{eq:upper_bound_average}
\end{multline}
In particular, we can find $\ell'\in(\ell(1-\alpha),\ell(1+\alpha))$, $\delta'\in(\delta(1-\alpha),\delta(1+\alpha))$ and an isometry $(\tau,R)\in \R^3\times SO(3)$ such that 
\begin{equation}
 E(\rho)\leq \sum_{z\in\Z^3}\sum_{j=1}^{24} E\Big(\chi_{\ell',\delta',j}(R\,\cdot\,-\ell' z-\tau)\rho\Big)
 +C\delta^2\log(\alpha^{-1})\int_{\R^3}\rho^2.
 \label{eq:upper_bound}
\end{equation}
\end{proposition}

The right side of~\eqref{eq:upper_bound_average} involves our incomplete partition of unity with holes, which is rotated (with the rotation $R$), translated (with the translation $\tau$) and dilated (with the dilation parameter $t$). The error is small only when $\delta$ (the size of the holes) is small. However we cannot take $\delta=0$ since that would make the gradient of the densities $\chi_{t\ell,t\delta,j}(R\,\cdot\,-t\ell z-\tau)\rho$ blow up. Nevertheless, the statement is that the energy decouples and the holes can be neglected, at the expense of an error of the order $\delta^2\int_{\R^3}\rho^2$. In~\eqref{eq:upper_bound_average} we use dilations for purely technical reasons, in order to better control error terms. 

\begin{proof}[Proof of Proposition~\ref{prop:upper_bound}]
Using~\eqref{eq:resolution_identity}, we write our density $\rho$ as follows
\begin{equation}
\rho(x)=\frac{1}{\ell^3}\int_{C_\ell}\sum_{z\in\Z^3}\sum_{j=1}^{24} \chi_{\ell,\delta,j}(x-\ell z-\tau)\rho(x)\,d\tau.
\end{equation}
For every fixed $\tau\in C_\ell$, we can construct a grand canonical trial state $\Gamma_\tau$ having the density 
$$\rho_{\Gamma_\tau}(x)=\sum_{z\in\Z^3}\sum_{j=1}^{24} \chi_{\ell,\delta,j}(x-\ell z-\tau)\,\rho(x).$$
For this we pick $\Gamma_\tau=\bigotimes_{j=1}^{24}\bigotimes_{z\in\Z^3} \Gamma_{\tau,z,j}$ where each $\Gamma_{\tau,z,j}$ has the density
$$\rho_{\Gamma_{\tau,z,j}}(x)=\chi_{\ell,\delta,j}(x-\ell z-\tau)\rho(x)$$
and minimizes the corresponding energy $E(\chi_{\ell,\delta,j}(\cdot-\ell z-\tau)\rho)$.
Since the quantum states $\Gamma_{\tau,z,j}$ have disjoint supports, we can anti-symmetrize the state $\Gamma_\tau$ in the standard manner. We denote by $\Gamma_{\tau,a}$ the anti-symmetrized state. The energy of $\Gamma_{\tau,a}$ is equal to that of $\Gamma_\tau$ and so is its density $\rho_{\Gamma_{\tau,a}}=\rho_{\Gamma_\tau}$.
Finally we take as trial state 
$$\Gamma=\frac{1}{\ell^3}\int_{C_\ell}\Gamma_{\tau,a}\,d\tau$$
which satisfies by construction that $\rho_\Gamma=\rho$.
We find the upper bound
\begin{align}
E(\rho)&\leq \frac{1}{\ell^3}\int_{C_{\ell}}E(\rho_{\Gamma_\tau})\,d\tau+\frac{1}{\ell^3}\int_{C_\ell}D(\rho_{\Gamma_\tau})\,d\tau-D(\rho)\nn\\
&=\frac{1}{\ell^3}\int_{C_\ell}\sum_{z\in\Z^3}\sum_{j=1}^{24} E\big(\chi_{\ell,\delta,j}(\cdot-\ell z-\tau)\rho\big)\,d\tau+\frac{1}{\ell^3}\int_{C_\ell}D(\rho_{\Gamma_\tau}-\rho)\,d\tau.
\label{eq:almost-final_upper_bound}
\end{align}
Here we employ the usual notation
\begin{equation}
D(\rho):=\frac12\int_{\R^3}\int_{\R^3}\frac{\rho(x)\rho(y)}{|x-y|}dx\,dy.
\label{eq:def_D_rho_rho}
\end{equation}
In the second line of~\eqref{eq:almost-final_upper_bound} we have used that the energy of a tensor product of states of disjoint supports is the sum of the energies of the pieces~\cite{LewLieSei-18}. This is because the cross terms in the direct energy exactly cancel with the many-particle interactions of different states in the tensor product. 

The error term in~\eqref{eq:almost-final_upper_bound} is solely due to the nonlinearity of the direct term and it may be rewritten as
\begin{equation*}
\frac{1}{\ell^3}\int_{C_\ell}D(\rho_{\Gamma_\tau}-\rho)\,d\tau\\
=\frac{1}{\ell^3}\int_{C_\ell}D\left(\sum_{z\in\Z^3}\sum_{j=1}^{24}(\1_{\ell\bDelta_j}-\chi_{\ell,\delta,j})(\cdot-\ell z-\tau)\rho\right)\,d\tau.
\end{equation*}
For every real-valued $(\ell\Z^3)$--periodic function $f$, we have
\begin{align*}
&\frac{1}{\ell^3}\int_{C_\ell}D\Big(f(\cdot-\tau)\rho\Big)\,d\tau\\
&\qquad =\frac12\sum_{k\in (2\pi/\ell)\Z^3} \left|\frac{1}{\ell^3}\int_{C_\ell}f(z)e^{-ik\cdot z}\,dz\right|^2\iint_{\R^3\times\R^3} \frac{e^{ik\cdot(x-y)}}{|x-y|}\rho(x)\,\rho(y)\,dx\,dy\\
&\qquad =2\pi\sum_{k\in (2\pi/\ell)\Z^3} \left|\frac{1}{\ell^3}\int_{C_\ell}f(z)e^{-ik\cdot z}\,dz\right|^2\int_{\R^3}\frac{|\widehat{\rho}(p)|^2}{|p-k|^2}\,dp.
\end{align*}
Hence we obtain 
\begin{equation*}
\frac{1}{\ell^3}\int_{C_\ell}D(\rho_{\Gamma_\tau}-\rho)\,d\tau
=2\pi\sum_{k\in 2\pi\Z^3} \left|\int_{C_1}f_{\delta/\ell}(z)e^{-ik\cdot z}\,dz\right|^2\int_{\R^3}\frac{|\widehat{\rho}(p)|^2}{|p-k/\ell|^2}\,dp,
\label{eq:almost_final_computation_error_term_upper_bound}
\end{equation*}
with
\begin{equation}
f_{\eps}(x)=\sum_{j=1}^{24}\left(\1_{\mu_j\bDelta}-\frac{1}{(1-\eps)^{3}}\1_{\mu_j(1-\eps)\bDelta}\ast \eta_\eps\right)
\label{eq:def_f_eps}
\end{equation}
for $\epsilon=\delta/\ell$. We have
$$\int_{C_1}f_{\delta/\ell}(z)\,dz=0.$$
Since all the functions appearing in the sum on the right side of~\eqref{eq:def_f_eps} are supported in the unit cube, we also obtain for $k\in 2\pi\Z^3\setminus\{0\}$,
\begin{equation}
\int_{C_1}f_{\delta/\ell}(z)e^{-ik\cdot z}\,dz=- \frac{(2\pi)^{3}}{(1-\eps)^{3}}\sum_{j=1}^{24}\widehat{\1_{\mu_j(1-\eps)\bDelta}}(k)\widehat{\eta_1}(\eps k).
\end{equation}
This results in the final formula for the error term
\begin{multline}
\frac{1}{\ell^3}\int_{C_\ell}D(\rho_{\Gamma_\tau}-\rho)\,d\tau\\
=(2\pi)^{7}\sum_{\substack{k\in 2\pi\Z^3\\ k\neq0}}\left|\widehat{\eta_1}(\eps k)\right|^2 \left|\frac{1}{(1-\eps)^{3}}\sum_{j=1}^{24}\widehat{\1_{\mu_j(1-\eps)\bDelta}}(k)\right|^2\int_{\R^3}\frac{|\widehat{\rho}(p)|^2}{|p-k/\ell|^2}\,dp.
\label{eq:final_computation_error_term_upper_bound}
\end{multline}

In order to control the denominator $|p-k/\ell|^2$, we are going to average our calculation over all the rotations of the tiling. We also replace $\ell$ and $\delta$ by, respectively,  $t\ell$ and $t\delta$ and we average over $t\in(1-\alpha,1+\alpha)$ with a weight $t^{-4}$. Rotating the tiling is the same as rotating $\rho$. In addition, $\epsilon=\delta/\ell$ is independent of $t$. Hence we are left with estimating 
\begin{align*}
&\left(\int_{1-\alpha}^{1+\alpha}\frac{dt}{t^4}\right)^{-1}\int_{1-\alpha}^{1+\alpha}\frac{dt}{t^4}\int_{SO(3)}dR\frac{1}{|p-Rk/(t\ell)|^2}\\
&\qquad\qquad=\frac{\ell^2}{4\pi|k|^2}\frac{3(1-\alpha^2)^3}{2\alpha(3+\alpha^2)}\int_{\frac{1}{1+\alpha}}^{\frac{1}{1-\alpha}}r^2dr\int_{S^2}d\omega\frac{1}{|p'-\omega r|^2}\\
&\qquad\qquad =\frac{\ell^2}{4\pi|k|^2}\frac{3(1-\alpha^2)^3}{2\alpha(3+\alpha^2)}\1_{A_\alpha}\ast\frac{1}{|\cdot |^2}(p')
\end{align*}
with $p'=p\ell/|k|$ and where $A_\alpha$ is the annulus
$$A_\alpha=\left\{\frac{1}{1+\alpha}\leq |x|\leq \frac{1}{1-\alpha}\right\}.$$

We will use the following estimate
\begin{lemma}\label{lem:estim_potential_annulus}
We have
\begin{equation}
\norm{\1_{A_\alpha}\ast\frac{1}{|\cdot |^2}}_{L^\ii}\leq C\alpha \log(\alpha^{-1})
\label{eq:estim_potential_annulus}
\end{equation}
for all $\alpha\leq1/2$. 
\end{lemma}

The proof of~\eqref{eq:estim_potential_annulus} is a simple computation which is provided at the very end of the proof.
Using~\eqref{eq:estim_potential_annulus} we obtain
\begin{multline}
\left(\int_{1-\alpha}^{1+\alpha}\frac{dt}{t^4}\right)^{-1}\int_{1-\alpha}^{1+\alpha}\frac{dt}{t^4}\int_{SO(3)}dR\frac{1}{(t\ell)^3}\int_{C_{t\ell}}D(\rho_{\Gamma_\tau,R}-\rho)\,d\tau\\
\leq C\log(\alpha^{-1})\left(\sum_{\substack{k\in 2\pi\Z^3\\ k\neq0}}\frac{\ell^2}{|k|^2}\left|\widehat{\eta_1}(\eps k)\right|^2 \left|\frac{1}{(1-\eps)^{3}}\sum_{j=1}^{24}\widehat{\1_{\mu_j(1-\eps)\bDelta}}(k)\right|^2\right)\int_{\R^3}\rho^2
\label{eq:very_final_computation_error_term_upper_bound}
\end{multline}
and it remains to estimate the sum in the parenthesis. For this we have to bound $\sum_{j=1}^{24}\widehat{\1_{\mu_j(1-\eps)\bDelta}}(k)$.

\begin{lemma}[Fourier transform of the reduced tetrahedra]\label{lem:Fourier_tetrahedron}
We have 
\begin{multline}
\left|\frac{1}{(1-\eps)^{3}}\sum_{j=1}^{24}\widehat{\1_{\mu_j(1-\eps)\bDelta}}(k)\right|^2\\
\leq C\left(\eps^4+\eps^2|k|^2\left|\int_{C_1}\bigg( x- \sum_{j=1}^{24}z_j\1_{\bDelta_j}\bigg)e^{-ik\cdot x}\,dx\right|^2\right)
\label{eq:estim_Fourier_reduced_tetrahedron}
\end{multline}
for all $0<\eps<1/2$ and all $k\in 2\pi\Z^3\setminus\{0\}$.
\end{lemma}

\begin{proof}
We recall that $\bDelta_j=R_j\bDelta_1-z_j$ with $\mu_j=(z_j,R_j)\in C_1\times SO(3)$.
We have
\begin{align*}
(1-\eps)^{-3}\widehat{\1_{\mu_j(1-\eps)\bDelta}}(k)&=(2\pi)^{-3/2}(1-\eps)^{-3}\int_{\mu_j(1-\eps)\bDelta}e^{-ik\cdot x}\,dx\\
&=(2\pi)^{-3/2}\int_{\bDelta}e^{-ik\cdot (R_j(1-\eps)x-z_j)}\,dx\\
&=(2\pi)^{-3/2}\int_{\mu_j\bDelta}e^{-ik\cdot x+i\eps k\cdot (x-z_j)}\,dx.
\end{align*}
Since $k\in 2\pi\Z^3\setminus\{0\}$, the integral vanishes at $\eps=0$ after summing over $j$. Inserting the derivative at $\eps=0$ yields
\begin{multline*}
(1-\eps)^{-3}\sum_{j=1}^{24}\widehat{\1_{\mu_j(1-\eps)\bDelta}}(k)
=i(2\pi)^{-3/2}\eps k\cdot\int_{C_1}\bigg( x- \sum_{j=1}^{24}z_j\1_{\bDelta_j}\bigg)e^{-ik\cdot x}\,dx\\
-\eps^2\sum_{j=1}^{24}\int_0^1(1-s)\,ds\int_{\bDelta_j}\big(k\cdot (x-z_j)\big)^2e^{-ik\cdot x+i\eps sk\cdot (x-z_j)}\,dx.
\end{multline*}
We claim that the second term is uniformly bounded with respect to $k$. Indeed, one integration by parts gives
\begin{align}
&\int_0^1(1-s)\,ds\int_{\bDelta_j}\big(k\cdot (x-z_j)\big)^2e^{-ik\cdot x+i\eps sk\cdot (x-z_j)}\,dx\nn\\
&\qquad=i\int_0^1\frac{1-s}{1-\eps s}\,ds\int_{\bDelta_j}k\cdot (x-z_j)\; (x-z_j)\cdot \nabla_x e^{-ik\cdot x+i\eps sk\cdot (x-z_j)}\,dx\nn\\
&\qquad=i\int_0^1\frac{1-s}{1-\eps s}\,ds\int_{\partial \bDelta_j}k\cdot (x-z_j)\; (x-z_j)\cdot n_j(x) e^{-ik\cdot x+i\eps sk\cdot (x-z_j)}\,dx\nn\\
&\qquad\qquad-4i\int_0^1\frac{1-s}{1-\eps s}\,ds\int_{\bDelta_j}k\cdot (x-z_j) e^{-ik\cdot x+i\eps sk\cdot (x-z_j)}\,dx.\label{eq:to_be_proved_lemma_simplices}
\end{align}
Here $n_j(x)$ is the normalized vector perpendicular to $\partial\bDelta_j$ pointing outwards. Integrating once more in the same manner (involving the edges of the faces of $\partial\bDelta_j$ for the first term), we see that~\eqref{eq:to_be_proved_lemma_simplices} is bounded uniformly in $k$, hence we obtain~\eqref{eq:estim_Fourier_reduced_tetrahedron}.
\end{proof}

Inserting~\eqref{eq:estim_Fourier_reduced_tetrahedron} in~\eqref{eq:very_final_computation_error_term_upper_bound}, we obtain the two error terms
\begin{multline*}
\eps^2\sum_{\substack{k\in 2\pi\Z^3\\ k\neq0}}\left|\widehat{\eta_1}(\eps k)\right|^2\left|\int_{C_1}\bigg( x- \sum_{j=1}^{24}z_j\1_{\bDelta_j}\bigg)e^{-ik\cdot x}\,dx\right|^2\\
\leq \eps^2(2\pi)^3\norm{x- \sum_{j=1}^{24}z_j\1_{\bDelta_j}}_{L^2(C_1)}^2=C\eps^2 
\end{multline*}
(using here that $\norm{\widehat{\eta_1}}_{L^\ii}\leq (2\pi)^{3/2}$) and 
$$\eps^4\sum_{\substack{k\in 2\pi\Z^3\\ k\neq0}}\frac{\left|\widehat{\eta_1}(\eps k)\right|^2}{|k|^2}\underset{\eps\to0}\sim \eps^3\int_{\R^3}\frac{\left|\widehat{\eta_1}(k)\right|^2}{|k|^2}\,dk.$$
Recalling that $\eps=\delta/\ell$, our final estimate on the averaged error is proportional to
$$\delta^2\log(\alpha^{-1})\left(1+\frac{\delta}\ell\right)\int_{\R^3}\rho^2.$$
In order to conclude the proof of Proposition~\ref{prop:upper_bound}, it remains to provide the

\begin{proof}[Proof of Lemma~\ref{lem:estim_potential_annulus}]
We have
\begin{align*}
\1_{A_\alpha}\ast\frac{1}{|\cdot |^2}(x)&=2\pi \int_{\frac1{1+\alpha}}^{\frac1{1-\alpha}}r^2dr\int_0^\pi \sin(\phi)d\phi\frac{1}{r^2+|x|^2-2r|x|\cos\phi}\\
&=\frac{\pi}{|x|} \int_{\frac1{1+\alpha}}^{\frac1{1-\alpha}}\log\left(\frac{r+|x|}{\big|r-|x|\big|}\right)r\,dr\\
&=\pi|x| \int_{\frac{1}{|x|(1+\alpha)}}^{\frac1{|x|(1-\alpha)}}\log\left(\frac{r+1}{\big|r-1\big|}\right)r\,dr.
\end{align*}
For $|x|\leq 1/5$ and $0<\alpha<1/2$, the integrand is bounded on the corresponding interval and we obtain 
$$\1_{A_\alpha}\ast\frac{1}{|\cdot |^2}(x)\leq C\alpha.$$
Similarly, for $|x|\geq 4$ and $0<\alpha<1/2$ the integrand can be estimated by $r^2$, which gives again
$$\1_{A_\alpha}\ast\frac{1}{|\cdot |^2}(x)\leq C\frac{\alpha}{|x|^2}\leq C\alpha.$$
Finally, for $1/5\leq |x|\leq 4$, we have 
$$\1_{A_\alpha}\ast\frac{1}{|\cdot |^2}(x)\leq C\alpha+C \int_{\frac{1}{|x|(1+\alpha)}}^{\frac1{|x|(1-\alpha)}}\left|\log\big|r-1\big|\right|\,dr.$$
The last integral is over an interval of length 
$$\frac1{|x|(1-\alpha)}-\frac{1}{|x|(1+\alpha)} = \frac{2\alpha}{|x|(1-\alpha^2)}\leq \frac{40}3\alpha.$$
The integral is maximum when the interval is placed at the divergence point $r=1$. So we have 
$$\int_{\frac{1}{|x|(1+\alpha)}}^{\frac1{|x|(1-\alpha)}}\left|\log\big|r-1\big|\right|\,dr\leq \int_{\max(0,1-\frac{40}3\alpha)}^{1+\frac{40}3\alpha}\left|\log\big|r-1\big|\right|\,dr\leq C\alpha\log(\alpha^{-1}).$$
\end{proof}

This concludes the proof of Proposition~\ref{prop:upper_bound}.
\end{proof}

\subsection{Lower bound in terms of local densities}\label{sec:lower_bound}
Next we turn to the lower bound. We are going to use the same tiling made of tetrahedra, with the difference that we do not insert any hole. Similarly to~\eqref{eq:def_chi_ell_tiling}, we introduce
\begin{equation}
 \boxed{\xi_{\ell,\delta,j}:=\1_{\ell\mu_j\bDelta}\ast \eta_\delta}
 \label{eq:def_xi_ell_tiling}
\end{equation}
which forms a smooth partition of unity, without holes,
$$\sum_{z\in\Z^3}\sum_{j=1}^{24} \xi_{\ell,\delta,j}(x-\ell z)=1.$$

\begin{proposition}[Lower bound in terms of local densities]\label{prop:lower_bound}
There exists a universal constant $C$ such that for any $\sqrt{\rho}\in H^1(\R^3)$ and any $\delta>0$ with $0<\delta/\ell<1/C$, we have
\begin{multline}
E(\rho)\geq \frac{1-{C\delta}/\ell}{\ell^3}\sum_{z\in\Z^3}\sum_{j=1}^{24}\int_{SO(3)}\int_{C_\ell} E\Big(\xi_{\ell,\delta,j}(R\,\cdot\,-\ell z-\tau)\rho\Big)dR\,d\tau\\
-\frac{C}{\ell}\int_{\R^3}\Big(\left(1+\delta^{-1}\right)\rho+\delta^3\rho^2\Big).
 \label{eq:lower_bound_average}
\end{multline}
In particular, we can find an isometry $(\tau,R)\in \R^3\times SO(3)$ such that 
\begin{multline}
 E(\rho)\geq \left(1-\frac{C\delta}\ell\right) \sum_{z\in\Z^3}\sum_{j=1}^{24}E\Big(\xi_{\ell,\delta,j}(R\,\cdot\,-\ell z-\tau)\rho\Big)\\
-\frac{C}{\ell}\int_{\R^3}\Big(\left(1+\delta^{-1}\right)\rho+\delta^3\rho^2\Big).
 \label{eq:lower_bound}
\end{multline}
\end{proposition}

\begin{proof}
For a state $\Gamma=\bigoplus_{n\geq0}\Gamma_n$ on Fock space (commuting with the particle number operator) and an interaction potential $w$, we introduce the simplified notation
\begin{equation}
\cC_w(\Gamma):= \sum_{n\geq2}\tr_{\gH^n}\left(\sum_{1\leq j<k\leq n}w(x_j-x_k)\right)\Gamma_n
 \label{eq:def_short_Coulomb}
\end{equation}
and
\begin{equation}
D_w(\rho):=\frac{1}{2}\int_{\R^3}\int_{\R^3}w(x-y)\rho(x)\rho(y)\,dx\,dy.
\end{equation}
For the Coulomb potential $w(x)=|x|^{-1}$ we simply use the notation $\cC(\Gamma)$ and $D(\rho)$. For the kinetic energy, we write
\begin{equation}
\cT(\Gamma):=\sum_{n\geq1} \tr_{\gH^n}\Bigg(-\sum_{j=1}^n\Delta_{x_j}\Bigg)\Gamma_n
\end{equation}
and finally denote by
$$\cE(\Gamma):=\cT(\Gamma)+\cC(\Gamma)-D(\rho_\Gamma).$$
the total energy, with the direct term subtracted.

The proof uses the well known fact that, for any interaction potential $w$ and any state $\Gamma$ on Fock space, we have
\begin{equation}
\cC_w(\rho)-D_w(\rho)\geq 
\begin{cases}
\dps-\frac{w(0)}{2}\int_{\R^3}\rho_\Gamma&\text{when $\widehat{w}\in L^1(\R^d)$ and $\widehat{w}\geq0$,}\\
\dps-\frac{\int_{\R^3}w}{2}\int_{\R^3}(\rho_\Gamma)^2&\text{when $w\in L^1(\R^d)$ and $w\geq0$.}
\end{cases}
\label{eq:general_method_lower_bound}
\end{equation}
For the first bound see for instance~\cite[Eq.~(4.8)]{LewLieSei-18}. The second bound uses only that $\cC_w(\Gamma)\geq0$ and $D_w(\rho)\leq \norm{w}_{L^1}\norm{\rho}_{L^2}^2/2$, by Young's inequality. 

We are now ready to prove~\eqref{eq:lower_bound}. The smeared Graf-Schenker inequality from~\cite[Lemma~6]{GraSch-95} states that the potential 
$$\tilde w_\ell(x)=\frac1{|x|}-\left(1-\frac{\kappa\delta}\ell\right)\frac{\tilde h_{\ell,\delta}(x)}{|x|}-\frac{\kappa\delta^2}{\ell}\frac{\tilde h_{\ell,\delta}(0)}{|x|\left(\delta+|x|\right)}$$
has a positive Fourier transform for all $\ell>\kappa\delta$, with $\tilde w_\ell(0)=-\kappa\ell^{-1}\tilde{h}_{\ell,\delta}(0)$, where 
\begin{align*}
\tilde h_{\ell,\delta}(x-y)&=\frac1{|\ell\bDelta|}\int_{SO(3)}\big(\1_{\ell\bDelta}\ast\eta_\delta\big)\ast\big(\1_{-\ell\bDelta}\ast\eta_\delta\big)(Rx-Ry)\,dR\\
&=\frac1{|\ell\bDelta|}\int_{SO(3)}\int_{\R^3}\big(\1_{R^{-1}\ell \bDelta+z}\ast\eta_\delta\big)(y)\big(\1_{R^{-1}\ell \bDelta+z}\ast\eta_\delta\big)(x)\,dz\,dR.
\end{align*}
Here $\bDelta$ is a tetrahedron and $\kappa>0$ is a large enough constant. 
In addition, we have from~\cite[Proof of Lem.~5.5]{LewLieSei-18} that the potential
$$\tilde W_\delta(x)=\frac{1}{|x|(\delta+|x|)}-\frac{e^{-\frac{\sqrt{2}}\delta |x|}}{\delta|x|}$$
is positive and has positive Fourier transform, with $\tilde W_\delta(0)=\sqrt{2}\delta^{-2}$. 

Arguing exactly as in~\cite[Lem.~5.5]{LewLieSei-18} using~\eqref{eq:general_method_lower_bound}, we find that for any fermionic grand-canonical mixed state $\Gamma=\oplus_{n\geq0}\Gamma_n$ with density $\rho_\Gamma$, we have
\begin{align}
\cC(\Gamma)-D(\rho_\Gamma)\geq&\frac{1-{\kappa\delta}/\ell}{\ell^3}\sum_{z\in\Z^3}\sum_{j=1}^{24}\int_{SO(3)}\int_{C_\ell} \bigg\{\cC\left(\Gamma_{|\sqrt{\xi_{\ell,\delta,j}(R\,\cdot\,-\ell z-\tau)}}\right)\nn\\
&\qquad\qquad-D\Big(\xi_{\ell,\delta,j}(R\,\cdot\,-\ell z-\tau)\rho_\Gamma\Big)\bigg\}\,dR\,d\tau\nn\\
&-\frac{C}{\ell}\int_{\R^3}\rho_\Gamma-\frac{C\delta^3}{\ell}\int_{\R^3}(\rho_\Gamma)^2.
\label{eq:smeared_Graf_Schenker_Coulomb}
\end{align}
Here $\Gamma_{|f}$ is the geometrically $f$--localized state on Fock space~\cite{DerGer-99,HaiLewSol_2-09,Lewin-11}, that is, the unique state which has the $k$-particle reduced density matrices $f^{\otimes k}\Gamma^{(k)}f^{\otimes k}$. The last term proportional to $\delta^3/\ell$ comes from the $L^1$ norm of  $\ell^{-1}\delta e^{-\frac{\sqrt{2}}\delta |x|}|x|^{-1}$.

For the kinetic energy we use the IMS formula as in~\cite{GraSch-95,HaiLewSol_2-09} and~\cite[Lem.~5.6]{LewLieSei-18}, which yields an error in the form
$$\frac{N}{\ell^3}\int_{\R^3}|\nabla\sqrt{\xi_{\ell,\delta,j}}|^2=O\left(\frac{N}{\ell\delta}\right),$$
where $N=\int_{R^3}\rho_\Gamma$. For the total energy we obtain
\begin{multline}
\cE(\Gamma)\geq\frac{1-{\kappa\delta}/\ell}{\ell^3}\sum_{z\in\Z^3}\sum_{j=1}^{24}\int_{SO(3)}\int_{C_\ell} \cE\left(\Gamma_{|\sqrt{\xi_{\ell,\delta,j}(R\,\cdot\,-\ell z-\tau)}}\right)dR\,d\tau\\
-\frac{C}{\ell}\left(1+\frac{1}{\delta}\right)\int_{\R^3}\rho_\Gamma-\frac{C\delta^3}{\ell}\int_{\R^3}(\rho_\Gamma)^2,
\label{eq:smeared_Graf_Schenker_total_energy}
\end{multline}
which yields the result.
\end{proof}

\subsection{A convergence rate for tetrahedra}\label{sec:CV_rate_tetrahedra}

In this section we study the convergence of the energy per unit volume for tetrahedra and find a convergence rate. We introduce the energy per unit volume of a tetrahedron at constant density $\rho_0>0$
\begin{equation}
e_\bDelta(\rho_0,\ell,\delta):=|\ell\bDelta|^{-1}E\Big(\rho_0\,\1_{\ell\bDelta}\ast \eta_\delta\Big)
\label{eq:def_energy_tetrahedron}
\end{equation}
where $\eta_\delta(x)=(10/\delta)^{3}\eta_1(10x/\delta)$ with $\eta_1$ a fixed $C^\ii_c$ non-negative radial function with support in the unit ball and such that $\int_{\R^3}\eta_1=1$. We prove the following

\begin{proposition}[Thermodynamic limit for tetrahedra]\label{prop:estimates_tetrahedra}
For every fixed $\rho_0>0$, we have
\begin{equation}
\lim_{\substack{\delta/\ell\to0\\\delta^3/\ell\to0\\\ell\delta\to\ii}}e_\bDelta(\rho_0,\ell,\delta)=e_{\rm UEG}(\rho_0).
\label{eq:CV_energy_tetrahedron}
\end{equation}
For $\delta\leq \ell/C$ and $0<\alpha<1/2$, we have the upper bound
\begin{equation}
e_\bDelta(\rho_0,\ell,\delta)\leq e_{\rm UEG}(\rho_0)+C\frac{\rho_0}{\ell}\left(1+\delta^{-1}+\delta^3\rho_0+\delta\rho_0^{2/3}\right),
\label{eq:upper_bound_tetrahedron}
\end{equation}
and the averaged lower bound
\begin{equation}
\left(\int_{1-\alpha}^{1+\alpha}\frac{ds}{s^4}\right)^{-1}\int_{1-\alpha}^{1+\alpha}e_\bDelta(\rho_0,t\ell,t\delta)\,\frac{dt}{t^4}
\geq e_{\rm UEG}(\rho_0)-C\delta^2\rho_0^2\log(\alpha^{-1}).
\label{eq:lower_bound_tetrahedron}
\end{equation}
If in addition $\rho_0^{1/3}\ell\geq C$, we have the pointwise lower bound
\begin{equation}
 e_\bDelta(\rho_0,\ell,\delta)\geq e_{\rm UEG}(\rho_0) - C\delta\frac{\rho_0^{5/3}+\rho_0^{4/3}}{\ell}-C\frac{\rho_0^{23/15}+\rho_0^{18/15}}{\ell^{2/5}}.
 \label{eq:lower_bound_tetrahedron_worse}
\end{equation}
The constant $C$ only depends on the chosen regularizing function $\eta_1$. It is independent of $\rho_0,\ell,\delta,\alpha$.
\end{proposition}

We will later see that the condition $\delta^3/\ell\to0$ is actually not needed in the limit~\eqref{eq:CV_energy_tetrahedron}. It is an interesting problem to replace the error term in the lower bound~\eqref{eq:lower_bound_tetrahedron_worse} by an error similar to the upper bound~\eqref{eq:upper_bound_tetrahedron}. Note that the error term in~\eqref{eq:lower_bound_tetrahedron} goes to zero only when $\delta\to0$ whereas~\eqref{eq:lower_bound_tetrahedron_worse} does not require $\delta\to0$.

\begin{proof}
For fixed $\rho_0>0$ and $\delta>0$, the existence of the limit~\eqref{eq:CV_energy_tetrahedron} for $\ell\to\ii$ was proved in~\cite{LewLieSei-18}, using a lower bound similar to~\eqref{eq:lower_bound}. 

We consider a large tetrahedron $\ell'\bDelta$, smeared at a scale $\delta'$ and a tiling of smaller tetrahedra of size $\ell\ll\ell'$, smeared at scale $\delta$. Applying our lower bound~\eqref{eq:lower_bound_average}, we find
\begin{multline*}
e_\bDelta(\rho_0,\ell',\delta')\\
\geq \frac{1-C\frac{\delta}\ell}{|\ell'\bDelta|\ell^3}\sum_{z\in\Z^3}\sum_{j=1}^{24}\int_{C_\ell}\int_{SO(3)}E\Big(\rho_0\,(\1_{R(\ell\mu_j\bDelta-\ell z-\tau)}\ast \eta_\delta)\,(\1_{\ell'\bDelta}\ast \eta_{\delta'})\Big)dR\,d\tau\\
-\frac{C\rho_0}{\ell}\Big(1+\delta^{-1}+\delta^3\rho_0\Big)
\end{multline*}
where for the error term we have used that 
$$\int_{\R^3}(\rho_0\1_{\ell'\bDelta}\ast \eta_{\delta'})^2=\rho_0^2\int_{\R^3}(\1_{\ell'\bDelta}\ast \eta_{\delta'})^2\leq \rho_0^2\int_{\R^3}\1_{ \ell'\bDelta}\ast \eta_{\delta'}=\rho_0^2|\ell'\bDelta|.$$
For all the tetrahedra such that $R(\ell\mu_j\bDelta-\ell z-\tau) +B_{\delta/10}\subset (\ell'-\delta')\bDelta$, we obtain exactly $|\ell\bDelta|\,e_\bDelta(\rho_0,\ell,\delta)$ in the integral. The other tetrahedra are at a distance proportional to $\ell+\delta+\delta'$ from the boundary of $\ell'\bDelta$. Hence, using our lower bound~\eqref{eq:lower_bound_E} on the energy, they give rise to an error term of the order $\rho_0^{4/3}(\ell+\delta+\delta')/\ell'$. We obtain
\begin{multline}
e_\bDelta(\rho_0,\ell',\delta')\geq \left(1-C\sigma\frac{\delta}\ell-C\sigma\frac{\ell+\delta+\delta'}{\ell'}\right)e_\bDelta(\rho_0,\ell,\delta)\\-C\rho_0^{4/3}\frac{\ell+\delta+\delta'}{\ell'}
-\frac{C\rho_0}{\ell}\Big(1+\delta^{-1}+\delta^3\rho_0\Big).
\label{eq:lower_bound_tetrahedra_two_scales}
\end{multline}
Here $\sigma=1$ if $e_\bDelta(\rho_0,\ell,\delta)\geq0$ and $\sigma=0$ otherwise.

After taking the limit $\ell'\to\ii$ at fixed $\ell,\delta,\delta',\rho_0$, we obtain 
\begin{equation*}
e_{\rm UEG}(\rho_0)\geq \left(1-C\sigma\frac{\delta}\ell\right)e_\bDelta(\rho_0,\ell,\delta) -\frac{C\rho_0}{\ell}\Big(1+\delta^{-1}+\delta^3\rho_0\Big).
\end{equation*}
It follows from our upper bound~\eqref{eq:upper_bound_E} (see also~\cite[Rmk.~5.4]{LewLieSei-18}) that 
$$e_{\rm UEG}(\rho_0)\leq q^{-2/3}c_{\rm TF}\rho_0^{5/3}$$
for all $\rho_0>0$. Hence after dividing by $1-C\sigma\delta/\ell$  we have shown the claimed upper bound
\begin{equation*}
e_\bDelta(\rho_0,\ell,\delta)\leq e_{\rm UEG}(\rho_0)+C\frac{\rho_0}{\ell}\left(1+\delta^{-1}+\delta^3\rho_0+\delta\sigma\rho_0^{2/3}\right).
\end{equation*}
We may use exactly the same argument using our upper bound~\eqref{eq:upper_bound_average} in place of~\eqref{eq:lower_bound} and we obtain the lower bound~\eqref{eq:lower_bound_tetrahedron}.

Next we replace $\ell$ by $t\ell$ and $\delta$ by $t\delta$ in our lower bound~\eqref{eq:lower_bound_tetrahedra_two_scales}  on the energy of the large simplex of size $\ell'$, and  average over $t\in(1/2,3/2)$ with the measure $t^{-4}$. We then insert our lower bound~\eqref{eq:lower_bound_tetrahedron} and, after collecting the different error terms, we obtain
\begin{multline*}
e_\bDelta(\rho_0,\ell',\delta')\geq e_{\rm UEG}(\rho_0) - C\frac{\ell+\delta+\delta'}{\ell'}(\rho_0^{5/3}+\rho_0^{4/3})\\
-C\frac{\rho_0}{\ell}\left(1+\delta^{-1}+\delta^3\rho_0+\delta\rho_0^{2/3}\right)-C\delta^2\rho_0^2.
\end{multline*}
It is natural to choose $\delta=\ell^{-1/3}(\rho_0)^{-4/9}$ which provides the estimate
\begin{multline*}
e_\bDelta(\rho_0,\ell',\delta')\geq e_{\rm UEG}(\rho_0) - C\frac{\ell+\ell^{-1/3}(\rho_0)^{-4/9}+\delta'}{\ell'}(\rho_0^{5/3}+\rho_0^{4/3})\\
-C\frac{\rho_0^{13/9}+\rho_0^{10/9}}{\ell^{2/3}}-C\frac{\rho_0^{11/9}}{\ell^{4/3}}-C\frac{\rho_0^{2/3}}{\ell^2}
\label{eq:lower_bound_to_be_optimized_over_once_again}
\end{multline*}
and then $\ell=(\ell')^{3/5}\rho_0^{-2/15}$ which gives
\begin{equation*}
e_\bDelta(\rho_0,\ell',\delta')\geq e_{\rm UEG}(\rho_0) - C\frac{\delta'}{\ell'}(\rho_0^{5/3}+\rho_0^{4/3})-C\frac{\rho_0^{23/15}+\rho_0^{18/15}}{(\ell')^{2/5}}
\label{eq:lower_bound_optimized}
\end{equation*}
under the assumption that $\ell'(\rho_0)^{1/3}\geq C$. This is exactly~\eqref{eq:lower_bound_tetrahedron_worse}. The two bounds~\eqref{eq:upper_bound_tetrahedron} and~\eqref{eq:lower_bound_tetrahedron_worse} give the limit~\eqref{eq:CV_energy_tetrahedron}.
\end{proof}

\subsection{Proof of Theorem~\ref{thm:UEG_thermo_limit_quantum}}\label{sec:proof_thermo_limit}

Let $\Omega_N$ be a sequence of domains as in the statement, that is, such that $|\Omega_N|\to\ii$ and $|\partial\Omega_N+B_r|\leq Cr|\Omega_N|^{2/3}$ for all $r\leq |\Omega_N|^{1/3}/C$. Assume also that $\delta_N|\Omega_N|^{-1/3}\to0$. 

By following the proof of~\eqref{eq:lower_bound_tetrahedra_two_scales} we see that a similar inequality holds with the large tetrahedron $\ell'\bDelta$ replaced by $\Omega_N$. This gives
\begin{multline*}
\frac{E\big(\rho_0\1_{\Omega_N}\ast\eta_{\delta_N}\big)}{|\Omega_N|}\geq \left(1-C\sigma\frac{\delta}\ell-C\sigma\frac{\ell+\delta+\delta_N}{|\Omega_N|^{1/3}}\right)e_\bDelta(\rho_0,\ell,\delta)\\-C\rho_0^{4/3}\frac{\ell+\delta+\delta_N}{|\Omega_N|^{1/3}}
-\frac{C\rho_0}{\ell}\Big(1+\delta^{-1}+\delta^3\rho_0\Big). 
\end{multline*}
Under the sole condition that $\delta_N|\Omega_N|^{-1/3}\to0$, the right side tends to $e_{\rm UEG}(\rho_0)$ if we take for instance $\delta$ fixed and $\ell=|\Omega_N|^{1/6}$.

We then use the upper bound~\eqref{eq:upper_bound_average} with $\alpha=1/2$ as well as the fact that 
\begin{multline*}
E\Big(\rho_0\chi_{t\ell,t\delta,j}(R\,\cdot\,-t\ell z-\tau)\1_{\Omega_N}\ast\eta_{\delta_N}\Big)\\
\leq C\rho_0\ell^3\left(\rho_0^{2/3}+(\ell\delta)^{-1}\right)
+C\rho_0\int_{\R^3}\chi_{t\ell,t\delta,j}(R\,\cdot\,-t\ell z-\tau)\left|\nabla \sqrt{\1_{\Omega_N}\ast\eta_{\delta_N}}\right|^2
\end{multline*}
by~\eqref{eq:upper_bound_E}, for the tetrahedra close to the boundary. We find
\begin{multline}
 \frac{E\big(\rho_0\1_{\Omega_N}\ast\eta_{\delta_N}\big)}{|\Omega_N|}\\
 \leq  \left(1+C\sigma\frac{\ell+\delta+\delta_N}{|\Omega_N|^{1/3}}\right) \left(\int_{1/2}^{3/2}\frac{ds}{s^4}\right)^{-1}\int_{1/2}^{3/2}e_\bDelta(\rho_0,t\ell,t\delta)\frac{dt}{t^4}
 \\
 +C\frac{\ell+\delta+\delta_N}{|\Omega_N|^{1/3}}\rho_0\left(\rho_0^{2/3}+(\ell\delta)^{-1}\right)+\frac{C\rho_0}{\delta_N|\Omega_N|^{1/3}}+C\rho_0^2\delta^2.
\end{multline}
We have used here that 
$$\frac{1}{|\Omega_N|}\int_{\R^3}\left|\nabla \sqrt{\1_{\Omega_N}\ast\eta_{\delta_N}}\right|^2\leq \frac{C}{\delta_N|\Omega_N|^{1/3}}.$$
Under the additional assumption that $\delta_N|\Omega_N|^{1/3}\to\ii$, we may choose for instance $\ell=|\Omega_N|^{1/6}$ and $\delta=|\Omega_N|^{-1/12}$, which yields the result.\qed

\subsection{Replacing the local density by a constant density}\label{sec:deviation}

The goal of this section is to provide estimates on the variation of the energy in a (smeared) tetrahedron, when we replace the local density by a constant, chosen to be either the minimum or the maximum of the density in the tetrahedron.

\begin{proposition}[Replacing $\rho$ by a constant locally]\label{prop:replace_constant_density}
Let $p>3$ and $0<\theta<1$ such that 
\begin{equation}
 2\leq p\theta\leq 1+\frac{p}{2}.
 \label{eq:hyp_p_theta_proof}
\end{equation}
There exists a constant $C=C(p,\theta,q)$ such that, for $\ell\geq C$ and $\delta\leq \ell/C$, we have
\begin{align}
E\big(\rho\, (\1_{\ell\bDelta}\ast\eta_\delta)\big)\leq & E\left(\underline{\rho}\,(\1_{\ell\bDelta}\ast\eta_\delta)\right)+C\eps\int_{\R^3}\left(\rho+\rho^2\right)(\1_{\ell\bDelta}\ast\eta_\delta)\nn\\
&+C\int_{\R^3}\rho\left|\nabla\sqrt{\1_{\ell\bDelta}\ast\eta_\delta}\right|^2+\frac{C}\eps \int_{\R^3}|\nabla\sqrt\rho|^2(\1_{\ell\bDelta}\ast\eta_\delta)\nn\\
&+C\left(\frac{\ell^{2p}}{\eps^{p-1}}+\frac{\ell^p}{\eps^{\frac54 p-1}}\right)\int_{\ell\bDelta+B_\delta}|\nabla\rho^\theta|^p
\label{eq:replace_constant_rho_upper}
\end{align}
and
\begin{align}
 E\big(\rho\, (\1_{\ell\bDelta}\ast\eta_\delta)\big)\geq& E\left(\overline{\rho}\,(\1_{\ell\bDelta}\ast\eta_\delta)\right)-C\eps\ell^3\left(\overline\rho+\overline\rho^2\right)
-\frac{C\ell^2}{\delta}\overline{\rho}\nn\\
&-\frac{C}\eps \int_{\R^3}|\nabla\sqrt\rho|^2(\1_{\ell\bDelta}\ast\eta_\delta)\nn\\
&-C\left(\frac{\ell^{2p}}{\eps^{p-1}}+\frac{\ell^p}{\eps^{\frac54 p-1}}\right)\int_{\ell\bDelta+B_\delta}|\nabla\rho^\theta|^p
 \label{eq:replace_constant_rho_lower}
\end{align}
for all $0<\eps\leq1/2$, where
$$\underline{\rho} =\min_{x\in {\rm supp}(\1_{\ell\bDelta}\ast\eta_\delta)}\rho(x),\qquad \overline{\rho} =\max_{x\in {\rm supp}(\1_{\ell\bDelta}\ast\eta_\delta)}\rho(x)$$
are respectively the minimum and maximum value of $\rho$ on the support of $\1_{\ell\bDelta}\ast\eta_\delta$.
\end{proposition}

Under the assumption that $\int_{\ell\bDelta+B_\delta}|\nabla\rho^\theta|^p$ is finite, the density $\rho$ is continuous on $\overline{\ell\bDelta+B_\delta}$, so that $\underline\rho$ and $\overline\rho$ are well defined. 

We have already discussed in the beginning of Section~\ref{sec:upper_bound} the difficulty of deriving a subadditivity-type estimate relating $E(\rho_1+\rho_2)$ to $E(\rho_1)$ and $E(\rho_2)$. The following lemma provides a rather rough inequality, which however will be sufficient for our purposes.

\begin{lemma}[Rough subadditivity estimate]\label{lem:subadditivity_worse}
Let $\rho_1,\rho_2\in L^1(\R^3,\R_+)$ be two densities such that $\sqrt{\rho_1},\sqrt{\rho_2}\in H^1(\R^3)$. Then 
\begin{multline}
 E(\rho_1+\rho_2)\leq E(\rho_1)+C\eps\int_{\R^3}\left(\rho_1^{5/3}+\rho_1^{4/3}\right)+C\eps^{-2/3}\int_{\R^3}\rho_2^{5/3}\\
 +C\int_{\R^3}|\nabla\sqrt{\rho_2+\eps\rho_1}|^2+\frac{1-\eps}{\eps}D(\rho_2)
 \label{eq:subadditivity_worse}
\end{multline}
for all $0<\eps\leq1$. 
\end{lemma}

Here we have in mind that $\rho_2$ is small compared to $\rho_1$ and we estimate $E(\rho_1+\rho_2)$ in terms of $E(\rho_1)$ plus some error terms. The worse error in the estimate~\eqref{eq:subadditivity_worse} is $D(\rho_2)/\eps$, because it grows much faster than the volume. Later we will only use~\eqref{eq:subadditivity_worse} locally and this bad term will not be too large. But it will be responsible for the large power of $\eps$ in front of the gradient correction in our main estimate~\eqref{eq:LDA_main_estim}.
We conjecture that there is an inequality similar to~\eqref{eq:subadditivity_worse} without the term $D(\rho_2)/\eps$. 

Note that we can estimate 
$$\int_{\R^3}|\nabla\sqrt{\rho_2+\eps\rho_1}|^2\leq \int_{\R^3}|\nabla\sqrt{\rho_2}|^2+\eps\int_{\R^3}|\nabla\sqrt{\rho_1}|^2$$
by the convexity of $\rho\mapsto |\nabla\sqrt\rho|^2$.

\begin{proof}[Proof of Lemma~\ref{lem:subadditivity_worse}]
Fix an $\eps\in(0,1]$ and consider two optimal states $\Gamma_1$ and $\Gamma_2$ in Fock space, for $\rho_1$ and $\rho_2/\eps+\rho_1$, respectively. Then
$$\Gamma:=(1-\eps)\Gamma_1+\eps\Gamma_2$$
is a proper quantum state which has the density 
$$\rho_\Gamma=(1-\eps)\rho_1+\eps\left(\frac{\rho_2}{\eps}+\rho_1\right)=\rho_1+\rho_2.$$
Inserting this trial state and using~\eqref{eq:upper_bound_E} for $E(\rho_2/\eps+\rho_1)$, we deduce that 
\begin{multline*}
E(\rho_1+\rho_2)\leq (1-\eps) E(\rho_1)+\frac{C}{\eps^{2/3}}\int_{\R^3}\rho_2^{5/3}+C\int_{\R^3}|\nabla\sqrt{\rho_2+\eps\rho_1}|^2\\
+C\eps\int_{\R^3}\rho_1^{5/3}
-D(\rho_1+\rho_2)+(1-\eps)D(\rho_1)+\eps D(\rho_1+\rho_2/\eps).
\end{multline*}
We have
$$-D(\rho_1+\rho_2)+(1-\eps)D(\rho_1)+\eps D(\rho_1+\rho_2/\eps)=\frac{1-\eps}\eps D(\rho_2).$$
By the Lieb-Oxford inequality $E(\rho_1)\geq -C\int_{\R^3}\rho_1^{4/3}$, and the result follows.
\end{proof}

We are now able to provide the

\begin{proof}[Proof of Proposition~\ref{prop:replace_constant_density}]
We write $\rho=\underline{\rho}+(\rho-\underline{\rho})$ and apply~\eqref{eq:subadditivity_worse}. We obtain
\begin{align*}
E\big(\rho\, (\1_{\ell\bDelta}\ast\eta_\delta)\big)\leq & E\Big(\underline{\rho}\, (\1_{\ell\bDelta}\ast\eta_\delta)\Big)+C\eps\int_{\R^3}\left(\rho^{5/3}+\rho^{4/3}\right)(\1_{\ell\bDelta}\ast\eta_\delta)\\
&+\frac{C}{\eps^{2/3}}\int_{\R^3}(\rho-\underline{\rho})^{5/3}(\1_{\ell\bDelta}\ast\eta_\delta)+\frac1\eps D\Big((\rho-\underline{\rho})(\1_{\ell\bDelta}\ast\eta_\delta)\Big)\\
&+C\int_{\R^3}\left|\nabla\sqrt{(\1_{\ell\bDelta}\ast\eta_\delta)(\rho-(1-\eps)\underline{\rho})}\right|^2.
\end{align*}
In the first line we have used that $\underline{\rho}\leq\rho$ on the support of $\1_{\ell\bDelta}\ast\eta_\delta$ and that $\1_{\ell\bDelta}\ast\eta_\delta\leq 1$. First we can bound $\rho^{4/3}+\rho^{5/3}$ by $\rho+\rho^2$. Next, using
$$|\nabla \sqrt{fg}|^2=\frac{|\nabla (fg)|^2}{4fg}\leq \frac{f|\nabla g|^2}{2g}+\frac{g|\nabla f|^2}{2f},$$
and $\nabla(\rho-(1-\eps)\underline\rho)=\nabla\rho=2\sqrt\rho\nabla\sqrt\rho$, we can bound the gradient term pointwise by
\begin{multline*}
\left|\nabla\sqrt{(\1_{\ell\bDelta}\ast\eta_\delta)(\rho-(1-\eps)\underline{\rho})}\right|^2\\
\leq (\1_{\ell\bDelta}\ast\eta_\delta)\frac{2\rho\left|\nabla\sqrt{\rho}\right|^2}{\rho-(1-\eps)\underline\rho} 
+2\rho \left|\nabla \sqrt{\1_{\ell\bDelta}\ast\eta_\delta}\right|^2.
\end{multline*}
Since $\rho\geq\underline\rho$, we have $\eps\rho\leq \rho-(1-\eps)\underline\rho$ and hence
$$\frac{\rho}{\rho-(1-\eps)\underline\rho} \leq \frac{1}{\eps}.$$
This gives the estimate on the gradient term
\begin{multline*}
\int_{\R^3}\left|\nabla\sqrt{(\1_{\ell\bDelta}\ast\eta_\delta)(\rho-(1-\eps)\underline{\rho})}\right|^2\\
\leq \frac2\eps \int_{\R^3}|\nabla\sqrt\rho|^2(\1_{\ell\bDelta}\ast\eta_\delta)
+2\int_{\R^3}\rho \left|\nabla \sqrt{\1_{\ell\bDelta}\ast\eta_\delta}\right|^2.
\end{multline*}

Next we estimate the terms involving $\rho-\underline\rho$ in terms of the gradient of $\rho^\theta$. We use the Sobolev inequality in the bounded set $\ell\bDelta+B_\delta$ 
\begin{equation}
 \norm{u}^p_{L^\ii(\overline{\ell\bDelta+B_\delta})}\leq C\ell^{p-3} \int_{\ell\bDelta+B_\delta}|\nabla u(x)|^p\,dx
 \label{eq:Sobolev_domain}
\end{equation}
for $p>3$ and every continuous $u$ which vanishes at least at one point in $\ell\bDelta+B_\delta$ (we always assume $\delta\leq\ell/C$ so that $\ell\bDelta+B_\delta$ is included in a ball of radius proportional to $\ell$). 
By the Hardy-Littlewood-Sobolev inequality, this gives
\begin{align}
&D\Big((\rho-\underline{\rho})(\1_{\ell\bDelta}\ast\eta_\delta)\Big)\nn\\
&\qquad \qquad \leq C\norm{(\rho-\underline{\rho})(\1_{\ell\bDelta}\ast\eta_\delta)}_{L^{6/5}}^2\nn\\
&\qquad \qquad \leq  C\norm{\rho^\theta-\underline{\rho}^\theta}_{L^{\ii}(\ell\bDelta+B_\delta)}^{2}\left(\int_{\R^3}\rho^{\frac65(1-\theta)}(\1_{\ell\bDelta}\ast\eta_\delta)\right)^{\frac53}\nn\\
&\qquad \qquad \leq C\left(\ell^{2p}\int_{\ell\bDelta+B_\delta}|\nabla\rho^\theta|^p\right)^{\frac{2}p}\left(\int_{\R^3}\rho^{\frac{2p}{p-2}(1-\theta)}(\1_{\ell\bDelta}\ast\eta_\delta)\right)^{1-\frac{2}{p}}\nn\\
&\qquad \qquad \leq C\eps \left(\frac{\ell^{2p}}{\eps^{p-1}}\int_{\ell\bDelta+B_\delta}|\nabla\rho^\theta|^p+\eps \int_{\R^3}\rho^{\frac{2p}{p-2}(1-\theta)}(\1_{\ell\bDelta}\ast\eta_\delta)\right).\label{eq:estim_direct}
\end{align}
In the second estimate we have used that 
$$\rho-\underline\rho\leq C(\rho^\theta-\underline\rho^\theta)\rho^{1-\theta}$$ 
since $\theta\leq1$. In the third estimate we have used H\"older's inequality to obtain an integral to the power $1-2/p$. This yields some power of $\ell$ which has been taken into account in the first factor. In order to bound $\rho^{\frac{2p}{p-2}(1-\theta)}$ by $\rho+\rho^2$ we need that 
$$1\leq \frac{2p}{p-2}(1-\theta)\leq 2$$
which is equivalent to our assumption~\eqref{eq:hyp_p_theta_proof}.

Similarly, we can bound the other error term as follows
\begin{align}
&\int_{\R^3}(\rho-\underline{\rho})^{\frac53}(\1_{\ell\bDelta}\ast\eta_\delta)\nn\\
&\qquad \leq C\norm{\rho^\theta-\underline{\rho}^\theta}_{L^\ii(\ell\bDelta+B_\delta)}^{\frac53 a}\int_{\R^3}\rho^{\frac53(1-\theta a)}(\1_{\ell\bDelta}\ast\eta_\delta)\nn\\
&\qquad\leq C\left(\ell^{p}\int_{\ell\bDelta+B_\delta}|\nabla\rho^\theta|^p\right)^{\frac{5a}{3p}}\left(\int_{\R^3}\rho^{\frac{5p}{3p-5a}(1-\theta a)}(\1_{\ell\bDelta}\ast\eta_\delta)\right)^{1-\frac{5a}{3p}}\nn\\
&\qquad\leq C\eps^{\frac23}\left(\frac{\ell^{p}}{\eps^{\frac{p}a-1}}\int_{\ell\bDelta+B_\delta}|\nabla\rho^\theta|^p+\eps \int_{\R^3}\rho^{\frac{5p}{3p-5a}(1-\theta a)}(\1_{\ell\bDelta}\ast\eta_\delta)\right)
\label{eq:estim_L_5_3}
\end{align}
where $0< a\leq 1$ is a parameter to be chosen. As before we need the condition 
$$1\leq \frac{5p}{3p-5a}(1-\theta a)\leq 2$$
in order to bound the last term by $\rho+\rho^2$. This is equivalent to
$$2-\frac{p}{5a}\leq\theta p\leq 1+\frac{2p}{5a}$$
where the left inequality is always satisfied under our assumption~\eqref{eq:hyp_p_theta_proof}. If we choose $a=1$ then the upper bound on $p\theta$ is stronger than~\eqref{eq:hyp_p_theta_proof}. Hence we rather choose $a=4/5$ and obtain~\eqref{eq:replace_constant_rho_upper}.

The argument for~\eqref{eq:replace_constant_rho_lower} is similar. This time we write $\overline{\rho}=\rho+(\overline\rho-\rho)$ and obtain from~\eqref{eq:subadditivity_worse}
\begin{align*}
E\big(\overline\rho\, (\1_{\ell\bDelta}\ast\eta_\delta)\big)\leq & E\Big(\rho\, (\1_{\ell\bDelta}\ast\eta_\delta)\Big)+C\eps\int_{\R^3}\left(\rho^{5/3}+\rho^{4/3}\right)(\1_{\ell\bDelta}\ast\eta_\delta)\\
&+\frac{C}{\eps^{2/3}}\int_{\R^3}(\overline\rho-\rho)^{5/3}(\1_{\ell\bDelta}\ast\eta_\delta)+\frac1\eps D\Big((\overline\rho-\rho)(\1_{\ell\bDelta}\ast\eta_\delta)\Big)\\
&+C\int_{\R^3}\left|\nabla\sqrt{(\1_{\ell\bDelta}\ast\eta_\delta)(\overline\rho-(1-\eps)\rho)}\right|^2.
\end{align*}
The gradient term can be bounded above by
\begin{align*}
&\int_{\R^3}\left|\nabla\sqrt{(\1_{\ell\bDelta}\ast\eta_\delta)(\overline\rho-(1-\eps)\rho)}\right|^2\\
&\qquad \leq \frac2\eps \int_{\R^3}|\nabla\sqrt\rho|^2(\1_{\ell\bDelta}\ast\eta_\delta)
+2\overline\rho\int_{\R^3} \left|\nabla \sqrt{\1_{\ell\bDelta}\ast\eta_\delta}\right|^2\\
&\qquad \leq \frac2\eps \int_{\R^3}|\nabla\sqrt\rho|^2(\1_{\ell\bDelta}\ast\eta_\delta)
+\frac{C\ell^2}{\delta}\overline\rho.
\end{align*}
The other terms are estimated as before, using  that $\overline\rho-\rho\leq \overline\rho$.
\end{proof}

\subsection{Lipschitz regularity of $e_{\rm UEG}$}\label{sec:Lipschitz_UEG}

In this section we prove that the UEG energy $e_{\rm UEG}$ is locally Lipschitz. The main result is the following.

\begin{proposition}[Lipschitz regularity of $e_{\rm UEG}$]\label{prop:Lipschtiz}
There exists a universal constant $C$ so that 
\begin{equation}
e_{\rm UEG}(\rho)-C\big(\rho^{\frac13}+\rho^{\frac23}\big)\rho'\leq e_{\rm UEG}(\rho-\rho')\leq e_{\rm UEG}(\rho)+C\rho'\rho^{\frac13} 
\label{eq:estim_continuity_e_UEG}
\end{equation}
for every $0\leq \rho'\leq\rho$. In particular, we have
\begin{equation}
 \left|e_{\rm UEG}(\rho_1)-e_{\rm UEG}(\rho_2)\right|\leq C\left(\max(\rho_1,\rho_2)^{\frac13}+\max(\rho_1,\rho_2)^{\frac23}\right)|\rho_1-\rho_2|.
 \label{eq:Lipschitz}
\end{equation}
\end{proposition}

\begin{proof}
By scaling we have 
$$E\big(\alpha\rho(\alpha^{\frac13}\cdot)\big)=\min_\Gamma \left\{\alpha^{\frac23}\cT(\Gamma)+\alpha^{\frac13}(\cC(\Gamma)-D(\rho_\Gamma))\right\}\leq \alpha^{\frac13} E(\rho)$$
for $\alpha\leq1$. This proves that 
$$\frac{E\big(\alpha\rho_0\1_{\ell/\alpha^{\frac13} \bDelta}\ast\chi_{\delta/\alpha^{\frac13}}\big)}{\ell^3|\bDelta|}=\frac{e_\bDelta(\alpha\rho_0,\ell/\alpha^{\frac13},\delta/\alpha^{\frac13})}{\alpha}\leq \alpha^{\frac13}\,e_\bDelta(\rho_0,\ell,\delta).$$
Passing to the limit using Proposition~\ref{prop:estimates_tetrahedra}, we find
$$e_{\rm UEG}(\alpha\rho_0)\leq \alpha^{\frac43}e_{\rm UEG}(\rho_0)$$
for every $0\leq\alpha=1-\eps\leq1$, hence
$$e_{\rm UEG}\big((1-\eps)\rho_0\big)\leq (1-\eps)^{\frac43}e_{\rm UEG}(\rho_0)\leq e_{\rm UEG}(\rho_0)+C\eps e_{\rm UEG}(\rho_0)_-.$$
Here we have used the notation $x_-=\max(-x,0)$ for the negative part. Using that $e_{\rm UEG}(\rho_0)\geq -c_{\rm LO}\rho_0^{4/3}$ by the Lieb-Oxford inequality (with $c_{\rm LO}\leq1.64$), we obtain 
$$e_{\rm UEG}\big((1-\eps)\rho_0\big)\leq e_{\rm UEG}(\rho_0)+C\eps\rho_0^{\frac43}$$
for all $0\leq \eps\leq1$. This proves the upper bound in~\eqref{eq:estim_continuity_e_UEG}.

Similarly, we can write (still for $0\leq \alpha\leq1$)
\begin{align*}
&E\big(\alpha\rho(\alpha^{\frac13}\cdot)\big)+c_{\rm LO}\alpha^{\frac13}\int_{\R^3}\rho^{\frac43}\\
&\qquad =\min_\Gamma \left\{\alpha^{\frac23}\cT(\Gamma)+\alpha^{\frac13}\left(\cC(\Gamma)-D(\rho_\Gamma)+c_{\rm LO}\int_{\R^3}\rho^{\frac43}\right)\right\}\\
&\qquad \geq \alpha^{\frac23} \left(E(\rho)+ c_{\rm LO}\int_{\R^3}\rho^{\frac43}\right).
\end{align*}
This gives as before
$$e_{\rm UEG}(\alpha\rho_0)+\alpha^{\frac43}c_{\rm LO}\rho_0^{\frac43}\geq \alpha^{\frac53}\left(e_{\rm UEG}(\rho_0)+c_{\rm LO}\rho_0^{\frac43}\right).$$
Using this time $e_{\rm UEG}(\rho_0)\leq C\rho_0^{5/3}$, we obtain 
$$e_{\rm UEG}\big((1-\eps)\rho_0\big)\geq e_{\rm UEG}(\rho_0)-C\eps\big(\rho_0^{\frac43}+\rho_0^{\frac53}\big)$$
for all $0\leq \eps\leq1$.
\end{proof}

\subsection{Proof of Theorem~\ref{thm:LDA}}\label{sec:conclusion}

We have derived all the estimates we need to prove the main inequality~\eqref{eq:LDA_main_estim} in Theorem~\ref{thm:LDA}. 

Let $\rho\in L^1(\R^3,\R_+)\cap L^2(\R^3,\R_+)$ be any density so that $\nabla \sqrt{\rho}\in L^2(\R^3)$ and $\nabla\rho^\theta\in L^p(\R^3)$. First we recall from our upper bound~\eqref{eq:upper_bound_E} and the lower bound~\eqref{eq:lower_bound_E} that 
$$|E(\rho)|\leq c_{\rm TF}q^{-2/3}(1+\eps)\int_{\R^3}\rho^{5/3}+c_{\rm LO}\int_{\R^3}\rho^{4/3}+\frac{C(1+\eps)}{\eps}\int_{\R^3}|\nabla\sqrt\rho|^2.$$
Similarly, we have
\begin{equation}
 |e_{\rm UEG}(\rho)|\leq c_{\rm TF}q^{-2/3}\rho^{5/3}+c_{\rm LO}\rho^{4/3}.
 \label{eq:estim_simple_UEG}
\end{equation}
In particular, the inequality~\eqref{eq:LDA_main_estim} is obvious for large $\eps$ and we only have to consider small $\eps$. 

In our upper bound~\eqref{eq:upper_bound_average} and our lower bound~\eqref{eq:lower_bound_average}, the worse coefficient involving $\ell$ and $\delta$ in front of $\rho+\rho^2$ is 
$\delta^2+1/(\ell\delta)$.
This suggests to take
\begin{equation}
 \boxed{\delta=\sqrt\eps,\qquad \ell=\eps^{-\frac32}}
 \label{eq:choice_ell_delta}
\end{equation}
which we do for the rest of the proof. In fact, in our proof we will replace $\ell$ and $\delta$ by $t\ell$ and $t\delta$ and average over $t\in[1/2,3/2]$. 

\subsubsection*{Step 1. Upper bound}
Let us first take $1/4\leq \eps^{3/2}\ell\leq 2$ and $1/4\leq \delta\eps^{-1/2}\leq 2$ and derive an upper bound on $E(\rho(1_{\ell\bDelta}\ast\eta_\delta))$. We recall that $\bDelta$ is a tetrahedron of volume $1/24$ as described in Section~\ref{sec:upper_bound} and that $\eta_\delta(x)=(10/\delta)^{3}\eta_1(10x/\delta)$ with $\eta_1$ a fixed $C^\ii_c$ non-negative radial function with support in the unit ball and such that $\int_{\R^3}\eta_1=1$.
We denote by 
$$\underline\rho:=\min_{{\rm supp}(\1_{\ell\bDelta}\ast \eta_\delta)}\rho,\qquad \overline\rho:=\min_{{\rm supp}(\1_{\ell\bDelta}\ast \eta_\delta)}\rho$$
the maximal and minimal values of $\rho$ on the support of $\1_{\ell\bDelta}\ast \eta_\delta$, as in Proposition~\ref{prop:replace_constant_density}.
We use the upper bound~\eqref{eq:replace_constant_rho_upper} from Proposition~\ref{prop:replace_constant_density} which quantifies the error made when replacing $E(\rho(1_{\ell\bDelta}\ast\eta_\delta))$ by $E(\underline\rho(1_{\ell\bDelta}\ast\eta_\delta))$. For the latter we then use our estimate~\eqref{eq:upper_bound_tetrahedron} in Proposition~\ref{prop:estimates_tetrahedra} on the energy of a smeared tetrahedron. With our choice~\eqref{eq:choice_ell_delta} of $\ell$ and $\delta$ in terms of $\eps$, this leads to
\begin{align}
E\big(\rho\, (\1_{\ell\bDelta}\ast\eta_\delta)\big)\leq& e_{\rm UEG}(\underline{\rho})\int_{\R^3}\1_{\ell\bDelta}\ast\eta_\delta+C\eps\int_{\R^3}\left(\rho+\rho^2\right)(\1_{\ell\bDelta}\ast\eta_\delta)\nn\\
&\ +C\int_{\R^3}\rho\left|\nabla\sqrt{\1_{\ell\bDelta}\ast\eta_\delta}\right|^2\nn\\
&\ +\frac{C}\eps \int_{\R^3}|\nabla\sqrt\rho|^2(\1_{\ell\bDelta}\ast\eta_\delta) +\frac{C}{\eps^{4p-1}}\int_{\ell\bDelta+B_\delta}|\nabla\rho^\theta|^p.
\label{eq:replace_constant_rho_upper_proof1}
\end{align}
In the first line we have bounded the error terms in~\eqref{eq:upper_bound_tetrahedron}  by
$$\eps^{\frac32}\rho\left(1+\eps^{-\frac12}+\eps^{\frac32}\rho+\sqrt\eps \rho^{2/3}\right)\leq C\eps(\rho+\rho^2)$$
in order to to simplify our final bound. 
In the support of $\1_{\ell\bDelta}\ast\eta_\delta$ we have by~\eqref{eq:estim_continuity_e_UEG}
$$e_{\rm UEG}(\underline{\rho})\leq e_{\rm UEG}(\rho(x))+C(\rho(x)-\underline\rho)\rho(x)^{\frac13}$$
hence
\begin{equation*}
 e_{\rm UEG}(\underline{\rho})\int_{\R^3}\1_{\ell\bDelta}\ast\eta_\delta\leq\int_{\R^3}e_{\rm UEG}(\rho)\,(\1_{\ell\bDelta}\ast\eta_\delta)
 +C\int_{\R^3}(\rho-\underline\rho)\rho^{\frac13}(\1_{\ell\bDelta}\ast\eta_\delta).
\end{equation*}
Similarly as we did for~\eqref{eq:estim_L_5_3}, we can bound for $0<a\leq1$
\begin{align}
&\int_{\R^3}(\rho-\underline{\rho})\rho^{\frac13}(\1_{\ell\bDelta}\ast\eta_\delta)\nn\\
&\qquad \leq C\norm{\rho^\theta-\underline{\rho}^\theta}_{L^\ii(\ell\bDelta+B_\delta)}^{a}\int_{\R^3}\rho^{\frac43-\theta a}(\1_{\ell\bDelta}\ast\eta_\delta)\nn\\
&\qquad\leq C\left(\eps^{-\frac32p}\int_{\ell\bDelta+B_\delta}|\nabla\rho^\theta|^p\right)^{\frac{a}{p}}\left(\int_{\R^3}\rho^{\frac{4-3\theta a}{3}\frac{p}{p-a}}(\1_{\ell\bDelta}\ast\eta_\delta)\right)^{1-\frac{a}{p}}\nn\\
&\qquad\leq C\left(\frac{1}{\eps^{\frac32p+\frac{p}a-1}}\int_{\ell\bDelta+B_\delta}|\nabla\rho^\theta|^p+\eps \int_{\R^3}\rho^{\frac{4-3\theta a}{3}\frac{p}{p-a}}(\1_{\ell\bDelta}\ast\eta_\delta)\right).
\label{eq:bound_4_3}
\end{align}
Again we need 
$$1\leq \frac{4-3\theta a}{3}\frac{p}{p-a}\leq 2$$
which is equivalent to
$$2-\frac{2p}{3a}\leq\theta p\leq 1+\frac{p}{3a}$$
where the left side is automatically satisfied under our main assumption~\eqref{eq:hyp_p_theta}.
In order to get an error controlled by the other gradient terms, we need
$$\frac32p+\frac{p}{a}-1\leq 4p-1$$
which requires $a\geq 2/5$. Taking $a=2/3$ provides the smallest power of $\eps$. Collecting our estimates, we have proved the following upper bound on the energy in a tetrahedron
\begin{align}
E\Big(\rho(\1_{\ell\bDelta}\ast \eta_\delta)\Big)\leq& \int_{\R^3}e_{\rm UEG}\big(\rho(x)\big)\, (\1_{\ell\bDelta}\ast \eta_\delta)\,dx+C\eps\int_{\R^3}(\1_{\ell\bDelta}\ast \eta_\delta)\big(\rho+\rho^2\big)\nn\\
&+\frac{C}{\eps^{4p-1}}\int_{\ell\bDelta+B_\delta}|\nabla\rho^\theta|^p
+\frac{C}{\eps}\int_{\R^3}(\1_{\ell\bDelta}\ast \eta_\delta)|\nabla\sqrt\rho|^2\nn\\
&+C\int_{\R^3}\rho\left|\nabla\sqrt{\1_{\ell\bDelta}\ast \eta_\delta}\right|^2.
\label{eq:estim_local_upper}
\end{align}
Here we have considered a tetrahedron placed at the origin for simplicity, but we of course get a similar inequality for any tetrahedron, by translating and rotating $\rho$. 

Next we recall our upper bound~\eqref{eq:upper_bound_average} on the total energy $E(\rho)$
\begin{multline}
 E(\rho)\leq  \left(\int_{1/2}^{3/2}\frac{ds}{s^4}\right)^{-1}\int_{1/2}^{3/2}\frac{dt}{t^4}\int_{SO(3)}dR\int_{C_{t\ell}}\frac{d\tau}{(t\ell)^3}\times\\
 \times\sum_{z\in\Z^3}\sum_{j=1}^{24} E\Big(\chi_{t\ell,t\delta,j}(R\,\cdot\,-t\ell z-\tau)\rho\Big)+C\eps \int_{\R^3}\rho^2,
 \label{eq:upper_bound_average_proof}
\end{multline}
with $ \chi_{\ell,\delta,j}:=(1-\eps^2)^{-3}\1_{\ell\mu_j(1-\eps^2)\bDelta}\ast \eta_\delta$. We also recall from Section~\ref{sec:upper_bound} that $\delta/\ell=(t\delta)/(t\ell)=\eps^2$.
Inserting~\eqref{eq:estim_local_upper} into~\eqref{eq:upper_bound_average_proof} and using the fact~\eqref{eq:resolution_identity} that $\chi_{t\ell,t\delta,j}$ forms a partition of unity after averaging over translations and rotations, we obtain 
\begin{multline}
 E(\rho)\leq  (1-\eps^2)^{3}\int_{\R^3}e_{\rm UEG}\Big((1-\eps^2)^{-3}\rho(x)\Big)\,dx+C\eps \int_{\R^3}\big(\rho+\rho^2\big)\\
 +\frac{C}{\eps}\int_{\R^3}|\nabla\sqrt\rho|^2+\frac{C}{\eps^{4p-1}}\int_{\R^3}|\nabla\rho^\theta|^p.
 \label{eq:upper_bound_average_proof2}
\end{multline}
Note that when we sum over the tiling, the sets $t\ell\mu_j\bDelta+B_{t\delta}$ have finitely many intersections,  which just results in a bigger constant in front of $|\nabla\rho^\theta|^p$. We have also used that 
\begin{multline*}
\int_{SO(3)}dR\int_{C_{t\ell}}\frac{d\tau}{(t\ell)^3} \sum_{z\in\Z^3}\sum_{j=1}^{24}\int_{\R^3}\rho\left|\nabla\sqrt{\chi_{t\ell,t\delta,j}(R\,\cdot\,-\ell z-\tau)}\right|^2\\
=\frac{24}{(t\ell)^3}\int_{\R^3}\rho\int_{\R^3}\left|\nabla\sqrt{\chi_{t\ell,t\delta,j}(R\,\cdot\,-\ell z-\tau)}\right|^2\leq \frac{C}{\ell\delta}\int_{\R^3}\rho= C\eps \int_{\R^3}\rho.
\end{multline*}
From Proposition~\ref{prop:Lipschtiz} and~\eqref{eq:estim_simple_UEG}, we have 
$$(1-\eps^2)^3\int_{\R^3}e_{\rm UEG}\Big((1-\eps^2)^{-3}\rho\Big)\leq \int_{\R^3}e_{\rm UEG}\big(\rho\big)+C\eps^2\int_{\R^3}\left(\rho^{\frac43}+\rho^{\frac53}\right)$$
hence we obtain the desired upper bound 
\begin{equation}
 E(\rho)\leq \int_{\R^3}e_{\rm UEG}\big(\rho\big)+ \eps \int_{\R^3}\big(\rho+\rho^2\big)
 +\frac{C}{\eps}\int_{\R^3}|\nabla\sqrt\rho|^2+\frac{C}{\eps^{4p-1}}\int_{\R^3}|\nabla\rho^\theta|^p
 \label{eq:proved_upper_bound}
\end{equation}
for $\eps$ small enough. 

\subsubsection*{Step 2. Lower bound}
The lower bound is slightly more tedious since all our lower estimates involve $\overline{\rho}$ which can in general not be bounded by $\rho$. We shall argue as follows. 
First we average our lower bound~\eqref{eq:lower_bound_average} over $t$. This gives
\begin{multline}
E(\rho)\geq \left(\int_{\frac12}^{\frac32}\frac{ds}{s^4}\right)^{-1}\int_{\frac12}^{\frac32}\frac{dt}{t^4}\frac{1-C\eps}{(t\ell)^3}\times\\
\times\sum_{z\in\Z^3}\sum_{j=1}^{24}\int_{SO(3)}\int_{C_{t\ell}} E\Big(\xi_{t\ell,t\delta,j}(R\,\cdot\,-t\ell z-\tau)\rho\Big)dR\,d\tau\\
-C\eps\int_{\R^3}\big(\rho+\eps^2\rho^2\big).
 \label{eq:lower_bound_average_proof}
\end{multline}
We recall that here $\xi_{\ell,\delta,j}=\1_{\ell\mu_j\bDelta}\ast\eta_\delta$, see Section~\ref{sec:lower_bound}. In order to use the same argument as for the upper bound, we are going to prove the estimate
\begin{multline}
\left(\int_{\frac12}^{\frac32}\frac{ds}{s^4}\right)^{-1}\int_{\frac12}^{\frac32}\frac{dt}{t^4(t\ell)^3}\bigg\{E\Big(\rho(\1_{t\ell\bDelta}\ast \eta_{t\delta})\Big)-\int_{\R^3}e_{\rm UEG}\big(\rho(x)\big)\, (\1_{t\ell\bDelta}\ast \eta_{t\delta})\,dx\\
+C\int_{\R^3}\rho|\nabla\sqrt{\1_{t\ell\bDelta}\ast \eta_{t\delta}}|^2+C\eps\int_{\R^3}\big(\rho+\rho^2\big)(\1_{t\ell\bDelta}\ast \eta_{t\delta})\\+\frac{C}{\eps}\int_{\R^3}|\nabla\sqrt\rho|^2(\1_{t\ell\bDelta}\ast \eta_{t\delta})\bigg\}
\geq-\frac{C}{\eps^{4p-1}}\int_{2\ell\bDelta+B_{2\delta}}|\nabla\rho^\theta|^p.
\label{eq:estim_local_lower}
\end{multline}
That the last integral is over the larger set $2\ell\bDelta+B_{2\delta}$ will only affect the multiplicative constant $C$. Inserting~\eqref{eq:estim_local_lower} into~\eqref{eq:lower_bound_average_proof} gives a bound as in~\eqref{eq:proved_upper_bound} but in the opposite direction. This concludes the proof of the theorem and it therefore only remains to prove~\eqref{eq:estim_local_lower}.

With an abuse of notation we consider the minimal and maximal values over the larger set $2(\ell\bDelta+B_\delta)$,
\begin{equation}
 \underline\rho:=\min_{2\ell\bDelta+B_{2\delta}}\rho,\qquad \overline\rho:=\min_{2\ell\bDelta+B_{2\delta}}\rho
 \label{eq:min_max_larger_set}
\end{equation}
instead of the corresponding definitions on the smaller set $\ell\bDelta+B_\delta$. 
First we again recall that, by~\eqref{eq:upper_bound_E} and~\eqref{eq:lower_bound_E}, we have 
\begin{align*}
&\left|E\Big(\rho(\1_{t\ell\bDelta}\ast\eta_{t\delta})\Big)-\int_{\R^3}(\1_{t\ell\bDelta}\ast\eta_{t\delta})e_{\rm UEG}\big(\rho(x)\big)\,dx\right|\\
&\qquad \leq C\int_{\R^3}(\1_{t\ell\bDelta}\ast\eta_{t\delta})\big(\rho^{4/3}+\rho^{5/3}\big)\\
&\quad\qquad +C\int_{\R^3}(\1_{t\ell\bDelta}\ast\eta_{t\delta})|\nabla\sqrt\rho|^2+C\int_{\R^3}\rho\left|\nabla\sqrt{\1_{t\ell\bDelta}\ast\eta_{t\delta}}\right|^2.
\end{align*}
Hence there is nothing to prove when 
\begin{equation*}
 \int_{\R^3}\big(\rho^{4/3}+\rho^{5/3}\big)(\1_{t\ell\bDelta}\ast\eta_{t\delta})\leq C\eps\int_{\R^3}\big(\rho+\rho^2\big)(\1_{t\ell\bDelta}\ast\eta_{t\delta})
 +\frac{1}{\eps^{4p-1}}\int_{2\ell\bDelta+B_{2\delta}}|\nabla\rho^\theta|^p.
\end{equation*}
This is the case if $\overline\rho^{1/3}\leq C\eps$, for instance. Hence we may assume in the following that $\overline\rho\geq C\eps^3$ and that 
\begin{equation*}
 \int_{\R^3}\big(\rho^{4/3}+\rho^{5/3}\big)(\1_{t\ell\bDelta}\ast\eta_{t\delta})\geq\frac{1}{\eps^{4p-1}}\int_{2\ell\bDelta+B_{2\delta}}|\nabla\rho^\theta|^p.
\end{equation*}
By~\eqref{eq:Sobolev_domain}, this implies 
$$\frac{\ell^3}{\eps^{3p\theta-4}}\overline\rho^{p\theta}\geq \frac{C}{\ell^{p-3}\eps^{4p-1}}\left(\overline\rho^\theta-\underline\rho^\theta\right)^p,$$
that is,
$$\overline\rho^\theta-\underline\rho^\theta\leq C\eps^{\frac{3}{p}\left(1+\frac{5p}{6}-\theta p\right)}\overline\rho^\theta.$$
Under our assumption~\eqref{eq:hyp_p_theta} on $p$ and $\theta$ the exponent is positive, hence we deduce that for $\eps$ small enough
$$\overline\rho\leq C\underline\rho\leq C\rho(x)$$
on $2\ell\bDelta+B_{2\delta}$. With this additional information we can use our previous estimates. 

By arguing exactly as in the proof of Proposition~\ref{prop:replace_constant_density} with $\overline\rho$ the maximum over $2\ell\bDelta+B_{2\delta}$ instead of the support of $\1_{t\ell\bDelta}\ast\eta_{t\delta}$, we get the estimate  similar to~\eqref{eq:replace_constant_rho_lower}
\begin{multline}
E\big(\rho\, (\1_{t\ell\bDelta}\ast\eta_{t\delta})\big)\geq E\left(\overline{\rho}\,(\1_{t\ell\bDelta}\ast\eta_{t\delta})\right)-C\eps\ell^3\left(\overline\rho+\overline\rho^2\right)\\
-\frac{C}\eps \int_{\R^3}|\nabla\sqrt\rho|^2(\1_{t\ell\bDelta}\ast\eta_{t\delta})
-\frac{C}{\eps^{4p-1}}\int_{2\ell\bDelta+B_{2\delta}}|\nabla\rho^\theta|^p.
 \label{eq:replace_constant_rho_lower_proof}
\end{multline}
From the fact that $\overline\rho\leq C\rho$, the second term on the right side can be bounded by 
$$C\eps\int_{\R^3}\big(\rho+\rho^2\big)(\1_{t\ell\bDelta}\ast\eta_{t\delta}).$$
Then we average over $t$ and use our lower estimate~\eqref{eq:lower_bound_tetrahedron} on the averaged energy of a tetrahedron. This gives
$$\left(\int_{\frac12}^{\frac32}\frac{ds}{s^4}\right)^{-1}\int_{\frac12}^{\frac32}\,\frac{dt}{t^4}\frac{E\left(\overline{\rho}\,(\1_{t\ell\bDelta}\ast\eta_{t\delta})\right)}{(t\ell)^3|\bDelta|}
\geq e_{\rm UEG}(\overline\rho)-C\eps\overline\rho^2.$$
The last term can again be bounded by 
$$\eps(t\ell)^{-3}\int_{\R^3}\rho^2(\1_{t\ell\bDelta}\ast\eta_{t\delta})$$ 
and included into the average over $t$.
Finally, using~\eqref{eq:estim_continuity_e_UEG} and $\overline\rho\leq C\rho$, we infer that 
$$e_{\rm UEG}(\overline\rho)\geq e_{\rm UEG}(\rho(x))-C\left(\overline\rho-\rho(x)\right)\rho(x)^{\frac13}$$
on the support of $\1_{t\ell\bDelta}\ast\eta_{t\delta}$. To conclude the proof of~\eqref{eq:estim_local_lower} we can proceed in the same way as for the upper bound~\eqref{eq:estim_local_upper}. This concludes the proof of Theorem~\ref{thm:LDA}.
\qed

\appendix
\section{Classical case}\label{app:classical}

In the classical case where the kinetic energy is neglected, the grand canonical energy functional is defined~\cite{LewLieSei-18} by
\begin{multline}
E_{\rm cl}(\rho):=\inf_{\substack{\sum_{n=0}^\ii\bP_n(\R^{3n})=1\\ \sum_{n=1}^\ii \rho_{\bP_n}=\rho}} \sum_{n=1}^\ii \int_{(\R^3)^n}\sum_{1\leq j<k\leq n}\frac{1}{|x_j-x_k|}d\bP_n(x_1,...,x_n)\\
-\frac12\int_{\R^3}\int_{\R^3}\frac{\rho(x)\rho(y)}{|x-y|}\,dx\,dy
\label{eq:def_classical_energy}
\end{multline}
where each $\bP_n$ is a symmetric probability measure on $(\R^3)^n$ with density 
$$\rho_{\bP_n}(x)=n\int_{\R^{3(n-1)}}d\bP_n(x,x_2...,x_n).$$
This classical energy~\eqref{eq:def_classical_energy} is obtained from the quantum energy in the limit 
$$\lim_{\alpha\to0}\alpha^{-\frac43}E(\alpha^3\rho(\alpha\cdot)\big)=E_{\rm cl}(\rho)$$
see~\cite{CotFriKlu-13,BinPas-17,Lewin-18,CotFriKlu-18}. When $\int_{\R^3}\rho=N\in\mathbb{N}$, the canonical version of $E_{\rm cl}$ reads
\begin{multline}
E^{\rm can}_{\rm cl}(\rho):=\inf_{\rho_{\bP}=\rho} \int_{(\R^3)^N}\sum_{1\leq j<k\leq N}\frac{1}{|x_j-x_k|}d\bP(x_1,...,x_N)\\
-\frac12\int_{\R^3}\int_{\R^3}\frac{\rho(x)\rho(y)}{|x-y|}\,dx\,dy.
\label{eq:def_classical_energy_canonical}
\end{multline}
In~\cite{LewLieSei-18} we have shown that for $\rho_N(x)=\rho_1(N^{-1/3}x)$ with $\int_{\R^3}\rho_1=1$,
\begin{equation}
\lim_{N\to\ii}\frac{E^{\rm can}_{\rm cl}(\rho_N)}{N}= \lim_{N\to\ii}\frac{E_{\rm cl}(\rho_N)}{N}=c_{\rm UEG}\int_{\R^3}\rho(x)^{\frac43}\,dx
\label{eq:limit_classical}
\end{equation}
where
$$c_{\rm UEG}=\lim_{\rho\to0^+}\frac{e_{\rm UEG}(\rho)}{\rho^{4/3}}<0$$
is the energy per unit volume of the classical uniform electron gas at density $1$. In this appendix we quickly explain how to derive the following quantitative estimate on the convergence rate in~\eqref{eq:limit_classical}.

\begin{theorem}[Estimate in the (grand-canonical) classical case]\label{thm:classical}
Let $p>3$ and $0<\theta<1$ such that $\theta p\geq4/3$.
There exists a universal constant $C=C(p,\theta)$ such that 
\begin{multline}
 \left|E_{\rm cl}(\rho)- c_{\rm UEG}\int_{\R^3}\rho(x)^{\frac43}\,dx\right|
 \leq \eps \int_{\R^3}\big(\rho(x)+\rho(x)^{\frac43}\big)\,dx\\
+\frac{C}{\eps^{b}}\int_{\R^3}|\nabla\rho^\theta(x)|^p\,dx
 \label{eq:LDA_main_estim_classical}
\end{multline}
with
$$b=\max\big\{2p-1,(1+3\theta)p-4\big\},$$
for every $\eps>0$ and every non-negative density $\rho\in L^1(\R^3)\cap L^{4/3}(\R^3)$ such that $\nabla \rho^\theta\in L^p(\R^3)$.
\end{theorem}

Under the condition that $\theta p\leq 1+p/3$, which is slightly more restrictive than in the quantum case~\eqref{eq:hyp_p_theta}, we get the much smaller power $2p-1$ of $\eps$ in front of the gradient term. Then, after optimizing~\eqref{eq:LDA_main_estim_classical} in $\eps$, we obtain the quantitative estimate
\begin{multline}
 \left|E_{\rm cl}(\rho_N)- N\,c_{\rm UEG}\int_{\R^3}\rho(x)^{\frac43}\,dx\right|\\
 \leq CN^{\frac56} \left(\int_{\R^3}\big(\rho(x)+\rho(x)^{\frac43}\big)\,dx\right)^{1-\frac1{2p}} \left(\int_{\R^3}|\nabla\rho^\theta(x)|^p\,dx\right)^{\frac1{2p}}
\end{multline}
for every $\rho_N(x)=\rho_1(N^{-1/3}x)$ and for $4/3\leq \theta p\leq 1+p/3$. The rate $N^{5/6}$ is better than the $N^{11/12}$ obtained in the quantum case, but still far from the expected rate $N^{1/3}$. 

\begin{proof}
The estimate~\eqref{eq:LDA_main_estim_classical} follows from the Lieb-Oxford inequality~\eqref{eq:Lieb-Oxford} (without the gradient term) when $\eps$ is large, so we only have to consider the case where $\eps$ is small. 

We use again the tiling~\eqref{eq:tiling} and, as in the proof in~\cite{LewLieSei-18}, the upper bound
\begin{equation}
E_{\rm cl}(\rho)\leq \sum_{z\in\Z^3}\sum_{j=1}^{24}E_{\rm cl}(\1_{\ell\mu_j\bDelta+\ell z}\rho)\leq \sum_{z\in\Z^3}\sum_{j=1}^{24}E_{\rm cl}(\1_{\ell\mu_j\bDelta+\ell z}\underline\rho)
\label{eq:upper_bound_classical}
\end{equation}
where $\underline\rho=\min_{\ell\mu_j\bDelta+\ell z}\rho$. The inequality~\eqref{eq:upper_bound_classical} is a consequence of the subadditivity and the negativity of $E_{\rm cl}$. Now it follows from~\cite[Cor.~3.4]{LewLieSei-18} and from the Graf-Schenker inequality as in Section~\ref{sec:CV_rate_tetrahedra}, that in a tetrahedron
$$c_{\rm UEG}\rho_0^{4/3}\leq \frac{E_{\rm cl}(\rho_0\1_{\ell\bDelta})}{\ell^3|\bDelta|}\leq c_{\rm UEG}\rho_0^{4/3}+\frac{C\rho_0}{\ell}.$$
This provides the upper bound
$$E_{\rm cl}(\rho)\leq c_{\rm UEG}\int_{\R^3}\rho^{4/3}+|c_{\rm UEG}|\sum_{z\in\Z^3}\sum_{j=1}^{24}\int_{\ell\mu_j\bDelta+\ell z}\left(\rho^{4/3}-\underline{\rho}^{4/3}\right) + \frac{C}{\ell}\int_{\R^3}\rho.$$
For the rest of the argument we use the notation $\eps=1/\ell$.
In each tetrahedron we can follow the argument in~\eqref{eq:bound_4_3} and estimate
\begin{align*}
&\int_{\R^3}(\rho^{\frac43}-\underline{\rho}^{\frac43})(\1_{\ell\bDelta}\ast\eta_\delta)\nn\\
&\qquad \leq C\norm{\rho^\theta-\underline{\rho}^\theta}_{L^\ii(\ell\bDelta+B_\delta)}^{a}\int_{\R^3}\rho^{\frac43-\theta a}(\1_{\ell\bDelta}\ast\eta_\delta)\nn\\
&\qquad\leq C\left(\frac{1}{\eps^p}\int_{\ell\bDelta+B_\delta}|\nabla\rho^\theta|^p\right)^{\frac{a}{p}}\left(\int_{\R^3}\rho^{\frac{4-3\theta a}{3}\frac{p}{p-a}}(\1_{\ell\bDelta}\ast\eta_\delta)\right)^{1-\frac{a}{p}}\nn\\
&\qquad\leq C\left(\frac{1}{\eps^{p+\frac{p}a-1}}\int_{\ell\bDelta+B_\delta}|\nabla\rho^\theta|^p+\eps \int_{\R^3}\rho^{\frac{4-3\theta a}{3}\frac{p}{p-a}}(\1_{\ell\bDelta}\ast\eta_\delta)\right)
\end{align*}
with $0<a\leq1$. In order to estimate the second term by $\rho+\rho^{4/3}$, we need that
$$1\leq \frac{4-3\theta a}{3}\frac{p}{p-a}\leq \frac43$$
which is equivalent to
$$\frac43\leq \theta p\leq 1+\frac{p}{3a}$$
and which we assume for the rest of the proof. 

We finally turn to the lower bound. Up to an appropriate rotation and translation of the tiling (or, equivalently, of the density $\rho$), the Graf-Schenker inequality gives the following lower bound~\cite[p.~100]{LewLieSei-18}
\begin{equation}
E_{\rm cl}(\rho)\geq \sum_{z\in\Z^3}\sum_{j=1}^{24}E_{\rm cl}(\1_{\ell\mu_j\bDelta+\ell z}\rho).
\label{eq:lower_bound_classical}
\end{equation}
We recall that 
$$E_{\rm cl}(\1_{\ell\mu_j\bDelta+\ell z}\rho)\geq -c_{\rm LO}\int_{\ell\mu_j\bDelta+\ell z}\rho^{\frac43}$$
by the Lieb-Oxford inequality~\eqref{eq:Lieb-Oxford}. We have nothing to prove in any tetrahedron $\ell\mu_j\bDelta+\ell z$ such that 
$$c_{\rm LO}\int_{\ell\mu_j\bDelta+\ell z}\rho^{\frac43}\leq \eps\int_{\ell\mu_j\bDelta+\ell z}\big(\rho+\rho^{\frac43}\big)
+\frac{A}{\eps^{p+\frac{p}{a}-1}}\int_{\ell\mu_j\bDelta+\ell z}|\nabla\rho^\theta|^p.$$
Here $A$ is a large constant to be chosen later. This is in particular the case when $\overline\rho^{1/3}=\max_{\ell\mu_j\bDelta+\ell z}\rho\leq \eps/c_{\rm LO}$.
So we may assume that 
$$c_{\rm LO}\int_{\ell\mu_j\bDelta+\ell z}\rho^{\frac43}> \eps\int_{\ell\mu_j\bDelta+\ell z}\big(\rho+\rho^{\frac43}\big)
+\frac{A}{\eps^{p+\frac{p}{a}-1}}\int_{\ell\mu_j\bDelta+\ell z}|\nabla\rho^\theta|^p$$
and that $\overline\rho>(\eps/c_{\rm LO})^3$. This implies
$$\overline\rho^\theta -\underline\rho^\theta\leq \frac{C\eps^{\frac{1}{a}-\frac1p}}{A^{\frac1p}}\overline\rho^{\frac4{3p}}\leq \frac{C\eps^{\frac3p(1+\frac{p}{3a}-\theta p)}}{A^{\frac1p}}\overline\rho^{\theta}.$$
The power of $\eps$ is non-negative when, again,
$$\theta p\leq 1+\frac{p}{3a}.$$
For $A$ large enough (or $\eps$ small enough) this gives $\overline\rho\leq C\underline\rho\leq C\rho$ in the simplex $\ell\mu_j\bDelta+\ell z$. The rest of the argument is then exactly the same as for the upper bound.

As a conclusion we obtain the bound~\eqref{eq:LDA_main_estim_classical} with the error term
$$\frac{C}{\eps^{p+\frac{p}{a}-1}}\int_{\ell\mu_j\bDelta+\ell z}|\nabla\rho^\theta|^p$$
and the restrictions that $p>3$, $0<\theta\leq1$, $0<a\leq1$ and
$$\frac43\leq\theta p\leq 1+\frac{p}{3a}.$$
In order to minimize the power of $\eps$ we want to take $a$ as large as possible, that is,
$$a=\min\left(1,\frac{p}{3(\theta p-1)}\right).$$
For $\theta p\leq 1+p/3$ we take $a=1$ whereas for $\theta p>1+p/3$ we choose the other value and get the stated inequality~\eqref{eq:LDA_main_estim_classical}. 
\end{proof}

\begin{remark}[Canonical case]\rm 
As was mentioned in~\eqref{eq:limit_classical}, in~\cite{LewLieSei-18} we could also handle the canonical case. We would easily obtain a quantitative estimate on the canonical energy $E^{\rm can}_{\rm cl}(\rho)$ if we knew the speed of convergence of $\ell^{-3}E_{\rm cl}^{\rm can}(\rho_0\1_{\ell\bDelta})$ to its limit $c_{\rm UEG}|\bDelta|$. Unfortunately, the argument used in~\cite[Lem.~3.2]{LewLieSei-18} to prove that the limit coincides with the grand canonical one does not seem to produce a quantitative bound. 
\end{remark}


\end{document}